\documentclass[a4paper,12pt]{article}
\usepackage{authblk}
\usepackage{fullpage}
\usepackage{ae,lmodern}
\usepackage[english,french]{babel}
\usepackage[utf8]{inputenc}  
\usepackage[T1]{fontenc}
\usepackage{url,csquotes}
\usepackage{float}
\usepackage{amsfonts,amssymb,enumerate}
\usepackage{amsmath}
\usepackage[shortlabels]{enumitem}
\allowdisplaybreaks[1]
\usepackage{amsthm}
\usepackage{graphicx}
\usepackage{booktabs}
\usepackage{siunitx}
\usepackage{subcaption}
\usepackage{bbm}
\usepackage{listings}
\usepackage{titling}
\usepackage{dsfont}
\usepackage{color}
\usepackage{mathrsfs}
\usepackage{enumitem}
\usepackage{bmpsize}
\usepackage{multibib}
\usepackage[hidelinks,hyperfootnotes=false]{hyperref}

 \theoremstyle{plain}   
 \newtheorem{thm}{Theorem}[section]
\newtheorem{lemma}[thm]{Lemma}
\newtheorem{coro}[thm]{Corollary}
\newtheorem{prop}[thm]{Proposition}

\numberwithin{equation}{section}

\title{Importance Sampling Enhanced with the COS Method for the Portfolio Risk Allocation}
\date{}

\author{
    Fang Fang\thanks{\href{https://fsquaredquant.nl}{FF Quant Advisory B.V.}, 3531 WR Utrecht, the Netherlands (\href{mailto:fang.fang@ffquant.nl}{fang.fang@ffquant.nl}) and Delft Institute of Applied Mathematics, Delft University of Technology, 2628 CD Delft, the Netherlands  ( \href{mailto:f.fang@tudelft.nl}{f.fang@tudelft.nl}).}
    \and    
    Xiaoyu Shen\thanks{\href{https://fsquaredquant.nl}{FF Quant Advisory B.V.}, 3531 WR Utrecht, the Netherlands (\href{mailto:xiaoyu.shen@ffquant.nl}{xiaoyu.shen@ffquant.nl}).} 
    \and
    Qinling Wang\thanks{Delft Institute of Applied Mathematics, Delft University of Technology, 2628 CD Delft, the Netherlands
    (\href{mailto:q.wang-7@tudelft.nl}{q.wang-7@tudelft.nl}).}
}

\date{\today}

\begin{document}

\selectlanguage{english}
\maketitle

\begin{abstract}
We study the calibration of importance-sampling proposals for rare portfolio
losses and obligor-level tail-risk contributions in multi-factor credit
models.  Cross-entropy importance sampling approximates the common-state
distribution conditional on a tail event, but its standard implementation
(CEIS) estimates the required moments from binary tail indicators and can be
unstable when few pilot paths reach an extreme threshold.  We propose ISCOS,
which replaces each indicator by a filtered COS approximation of
\(q_x(u)=\mathbb P(L\geq x\mid U=u)\).  Conditional on the common state,
defaults are independent and the portfolio-loss characteristic function is
available in closed form, so these smooth calibration weights can be computed
without an additional default simulation.  We derive Gaussian proposals for
the systematic factors and Gaussian--inverse-Gamma product proposals for
Student \(t\)-copulas.  With exact conditional probabilities, the ISCOS
raw-moment estimator is the Rao--Blackwellisation of CEIS and has no larger
covariance matrix.  If the filtered COS error is uniformly \(O(K^{-p})\), the
induced proposal parameters inherit the same order, yielding a fixed-threshold
calibration error of
\(O_{\mathbb P}(M_0^{-1/2})+O(K^{-p})\).  On an eleven-factor benchmark, exact
conditional block convolution confirms convergence of the Gaussian
COS-calibrated moments.  Under matched simulation budgets, ISCOS gives higher
event-weight effective sample sizes and shorter average nominal pointwise intervals
than CEIS in the reported Gaussian and Student \(t\)-copula configurations,
particularly for tail contributions, while end-to-end costs remain comparable.
The Gaussian comparison also includes the paper-aligned Glasserman's proposal.
\end{abstract}

\noindent\textbf{Keywords.}
Credit portfolio; importance sampling; cross-entropy method;
Fourier--cosine method; conditional Monte Carlo; rare events;
Gaussian copula; Student \(t\)-copula; risk allocation.

\begingroup
\small
\tableofcontents
\endgroup

\section{Introduction}
\label{sec:introduction}

Credit-portfolio risk measurement is a rare-event problem.  For a portfolio
loss
\[
L=\sum_{n=1}^{N} l_n Y_n,
\]
Value-at-Risk, threshold-based expected shortfall, and their obligor-level
allocations depend on events such as \(\{L\geq x\}\) and \(\{L=x\}\).  At a
confidence level of \(99.9\%\), an ordinary Monte Carlo sample of size \(M\)
contains only about \(M/1000\) tail observations on average, while the
exact-level event used for a VaR contribution may be much rarer.  Accurate
allocation can therefore require substantially more simulation than estimation
of an unconditional portfolio quantity.

Importance sampling makes these events more frequent under the simulation law
and corrects the result by a likelihood ratio.  In factor-copula credit models,
a natural strategy changes the distribution of the common factors and then
exponentially twists the conditional Bernoulli defaults.  Its performance
depends strongly on the factor proposal.  A proposal that is too mild produces
few useful tail observations; one that is too concentrated can generate
unstable likelihood ratios and may fail the moment conditions needed for
standard error estimates.

This paper studies how that proposal should be calibrated.  For a fixed loss
threshold \(x\), the natural target is the distribution of the common state
\(U\) conditional on \(A_x=\{L\geq x\}\), whose density is proportional to
\[
f_U(u)q_x(u),
\qquad
q_x(u)=\mathbb P(L\geq x\mid U=u).
\]
The cross-entropy projection of this distribution onto a Gaussian family is
therefore determined by tail-conditional factor moments.  Standard
cross-entropy importance sampling (CEIS) estimates those moments from the
binary indicators \(\mathbf 1_{\{L\geq x\}}\).  At an extreme threshold, most
pilot observations receive zero weight, making the fitted covariance noisy or
nearly singular.

We propose ISCOS, which replaces the indicator by a finite COS approximation
\(q_{x,K}(U)\) of the conditional tail probability.  Conditional on the common
state, obligor defaults are independent and the conditional loss
characteristic function is available analytically.  The COS method can thus
integrate out the idiosyncratic default simulation at the calibration stage.
Each pilot state is weighted by its chance of producing a tail loss, rather
than by the outcome of one conditional Bernoulli draw.  CEIS and ISCOS then use
the same weighted proposal fit and the same production importance-sampling
estimator.  ISCOS is therefore a proposal-calibration method, not a direct COS
estimator of the final risk contributions.

\subsection{Related Literature and Methodological Position}
\label{subsec:intro_literature}

Gaussian factor-copula models are standard in portfolio credit risk, while
Student \(t\)-copulas provide a heavy-tailed extension with tail dependence;
see, for example,
\cite{GlassermanLi2005,Glasserman2005,DemartaMcNeil2005}.  Marginal VaR and
expected-shortfall contributions can be represented as conditional expected
obligor losses and, under suitable conditions, provide full allocations of the
portfolio quantity \cite{Kalkbrener2005,Glasserman2005}.

Glasserman and Li~\cite{GlassermanLi2005} developed a two-stage rare-event
simulation method that combines a change of measure for the common factors
with conditional default twisting.  Glasserman~\cite{Glasserman2005} extended
this framework to marginal risk contributions and studied factor mean shifts,
covariance changes, and related asymptotic approximations.  Algorithms~6.1 and~6.2
use the paper-aligned Gaussian proposal $\mathcal N(\mu_x^{\mathrm{GL}},I)$,
where the mean is selected by a saddle-point optimisation.  This mean-only
proposal provides the Glasserman benchmark in our Gaussian experiment; the
Hessian covariance extension is not used in that comparison.

The cross-entropy method instead projects the formal zero-variance law onto a
tractable family by minimising Kullback--Leibler divergence
\cite{RubinsteinKroese2004,deBoerEtAl2005}.  For an exponential family, this is
a sufficient-statistic moment-matching problem.  Cross-entropy therefore gives
a global description of the tail-conditioned factor distribution, but its
usual indicator-based implementation remains vulnerable to a small number of
pilot tail hits.

The COS method recovers distributional quantities from characteristic
functions using Fourier--cosine coefficients \cite{FangOosterlee2008}.
Filtered discrete COS expansions, needed to control oscillations near jumps,
are analysed in \cite{ShenFangLiu2024}.  ISCOS uses this machinery only to
approximate \(q_x(u)\) during calibration.  The final CVaR and CES contributions
are still estimated from independent importance-sampling runs with the full
factor and conditional-default likelihood ratios.

\subsection{Contributions and Paper Organisation}
\label{subsec:intro_contributions}

The paper makes four contributions.

\begin{enumerate}[label=(\roman*)]
\item We formulate cross-entropy proposal construction in the low-dimensional
common-state space.  For the Gaussian component,
\[
\mu_x^\star=\mathbb E[Z\mid L\geq x],
\qquad
\Sigma_x^\star=\operatorname{Cov}(Z\mid L\geq x).
\]
For the Student \(t\)-copula, we derive a Gaussian--inverse-Gamma product
proposal; the scale component matches the tail-conditional sufficient
statistics \(\log W\) and \(W^{-1}\).

\item We introduce ISCOS calibration based on clipped COS weights
\(q_{x,K}(U)\).  The construction handles the left limit required by a
discrete non-strict tail event and supports filtering, batching, exact grouping
of homogeneous obligors, and covariance regularisation.

\item We separate the statistical and numerical errors.  With the exact
conditional probability, the ISCOS raw-moment estimator is the
Rao--Blackwellisation of its CEIS counterpart and has a no-larger covariance
matrix.  If
\[
\sup_u\lvert q_{x,K}(u)-q_x(u)\rvert=O(K^{-p}),
\]
then the induced proposal parameters inherit the same order, and the total
fixed-threshold calibration error is
\[
O_{\mathbb P}(M_0^{-1/2})+O(K^{-p}).
\]

\item We report Gaussian and Student \(t\)-copula experiments on the same
eleven-factor portfolio.  For the Gaussian model, exact conditional block
convolution validates the COS-calibrated moments, followed by a matched-budget
comparison of the Glasserman mean-shift proposal $\mathcal N(\mu_x^{\mathrm{GL}},I)$, CEIS, and ISCOS.  For the
Student \(t\)-copula, CEIS-t and ISCOS-t fit the same
Gaussian--inverse-Gamma family from a shared pilot.  In both reported
configurations, ISCOS gives higher event-weight effective sample sizes and
shorter average nominal pointwise intervals than CEIS, especially for tail
contributions, while the additional calibration cost remains small relative to
production simulation.  These numerical findings are descriptive and are not
presented as a universal variance ordering.
\end{enumerate}

The remainder of the paper introduces the portfolio models and risk-allocation
conventions, derives the population cross-entropy proposals, presents the CEIS
and ISCOS estimators, and analyses their calibration errors.  The final two
numerical sections treat the Gaussian and Student \(t\)-copula models before the
paper concludes.

\paragraph{Declaration on the use of generative AI} ChatGPT was used during the preparation of this manuscript to assist with language editing and to support parts of the research code. Anything generated by AI was reviewed by authors.

%

\section{Credit-Portfolio Models and Tail-Risk Quantities}
\label{sec:credit_portfolio_models}

We consider a portfolio observed over a fixed time horizon. Loss-given-default
amounts are deterministic, while dependence across obligors is generated by a
small number of common factors. The Gaussian and Student~\(t\) specifications
used below differ in their common state, but both make defaults conditionally
independent. This is the property that later allows us to evaluate conditional
tail probabilities with the COS method.

\subsection{Factor-copula loss model}
\label{subsec:portfolio_loss}

For obligor \(n\), let \(X_n\) be a latent creditworthiness variable,
\(\xi_n\) its default threshold, and \(l_n>0\) its loss-at-default. We use the
lower-tail convention
\begin{equation}
Y_n=\mathbf{1}_{\{X_n\leq\xi_n\}},
\qquad
L=\sum_{n=1}^{N}l_nY_n,
\label{eq:portfolio_loss}
\end{equation}
so that \(p_n=\mathbb{P}(X_n\leq\xi_n)\) is the unconditional default
probability. The portfolio loss has finite support in \([0,L_{\max}]\), where
\(L_{\max}=\sum_{n=1}^{N}l_n\).

Let \(Z\sim\mathcal{N}_d(0,I_d)\) denote the systematic factors and let
\(\varepsilon_1,\ldots,\varepsilon_N\) be independent standard normal
idiosyncratic factors, independent of \(Z\). Obligor \(n\) has loading vector
\(\beta_n\in\mathbb{R}^d\), with \(\|\beta_n\|_2^2<1\), and idiosyncratic
loading
\[
b_n=\sqrt{1-\|\beta_n\|_2^2}.
\]

\paragraph{Gaussian copula.}
The latent variable is
\[
X_n=\beta_n^\top Z+b_n\varepsilon_n,
\qquad
\xi_n=\Phi^{-1}(p_n),
\]
where \(\Phi\) is the standard normal CDF. Conditional on \(Z=z\), the default
probability is
\begin{equation}
p_n(z)
=
\Phi\left(\frac{\xi_n-\beta_n^\top z}{b_n}\right).
\label{eq:gaussian_conditional_default_probability}
\end{equation}

\paragraph{Student \(t\)-copula.}
Let \(V\sim\chi_\nu^2\) be independent of \(Z\) and the idiosyncratic factors,
and set
\[
W=\frac{\nu}{V}
\sim
\operatorname{InvGamma}\left(\frac{\nu}{2},\frac{\nu}{2}\right).
\]
We use the shape--scale parameterisation
\begin{equation}
f_W(w;\alpha,\beta)
=
\frac{\beta^\alpha}{\Gamma(\alpha)}
 w^{-\alpha-1}\exp\left(-\frac{\beta}{w}\right),
\qquad w>0.
\label{eq:inverse_gamma_density}
\end{equation}
The common scale \(W\) multiplies both the systematic and idiosyncratic parts:
\begin{equation}
X_n
=
\sqrt{W}\left(\beta_n^\top Z+b_n\varepsilon_n\right).
\label{eq:full_t_latent_variable}
\end{equation}
Then \(X_n\) has a standard Student \(t_\nu\) marginal,
\(\xi_n=T_\nu^{-1}(p_n)\), and
\begin{equation}
p_n(z,w)
=
\Phi\left(
\frac{\xi_n/\sqrt{w}-\beta_n^\top z}{b_n}
\right).
\label{eq:full_t_conditional_default_probability}
\end{equation}

It is convenient to write the common state as
\begin{equation}
U=
\begin{cases}
Z, & \text{Gaussian copula},\\[0.2em]
(Z,W), & \text{Student \(t\)-copula}.
\end{cases}
\label{eq:common_state_summary}
\end{equation}
In either model, we write \(p_n(u)=\mathbb{P}(Y_n=1\mid U=u)\) for the
corresponding conditional default probability.

\subsection{Conditional loss distribution}
\label{subsec:conditional_loss_distribution}

Conditional on \(U=u\), the indicators \(Y_1,\ldots,Y_N\) are independent
Bernoulli variables with success probabilities \(p_1(u),\ldots,p_N(u)\).
Hence the conditional characteristic function of the portfolio loss is
\begin{equation}
\varphi_{L\mid u}(\omega)
=
\mathbb{E}\!\left[e^{\mathrm{i}\omega L}\mid U=u\right]
=
\prod_{n=1}^{N}
\left[1+p_n(u)\left(e^{\mathrm{i}\omega l_n}-1\right)\right].
\label{eq:conditional_characteristic_function}
\end{equation}
This closed-form expression is the main model input used by ISCOS.

For a loss threshold \(x\), define
\begin{equation}
\mathcal{A}_x=\{L\geq x\},
\qquad
q_x(u)=\mathbb{P}(L\geq x\mid U=u).
\label{eq:conditional_tail_probability}
\end{equation}
Because \(L\) is discrete,
\[
q_x(u)=1-F_{L\mid u}(x^-),
\]
where \(F_{L\mid u}\) is the conditional CDF. Averaging over the common state
gives
\[
\mathbb{P}(L\geq x)=\mathbb{E}[q_x(U)].
\]
The COS approximation developed later is applied to
\eqref{eq:conditional_characteristic_function} to evaluate \(q_x(u)\).

\subsection{Tail-risk quantities and allocations}
\label{subsec:tail_risk_measures}

Let \(F_L(x)=\mathbb{P}(L\leq x)\). For \(\alpha\in(0,1)\),
\begin{equation}
\mathrm{VaR}_\alpha
=
\inf\{x\in\mathbb{R}:F_L(x)\geq\alpha\}.
\label{eq:var_definition}
\end{equation}
For a discrete loss distribution, the probability mass at
\(\mathrm{VaR}_\alpha\) need not vanish; in particular,
\[
\mathbb{P}(L>\mathrm{VaR}_\alpha)
\leq 1-\alpha
\leq \mathbb{P}(L\geq\mathrm{VaR}_\alpha).
\]
We therefore use the threshold-based convention
\[
\mathrm{ES}_\alpha
=
\mathbb{E}[L\mid L\geq\mathrm{VaR}_\alpha].
\]

The obligor-level contribution to VaR and the conditional
expected-shortfall contribution are defined by
\begin{align}
\mathrm{CVaR}_{n,\alpha}
&=
\mathbb{E}\!\left[
 l_nY_n
 \,\middle|\,
 L=\mathrm{VaR}_\alpha
\right],
\label{eq:var_contribution}\\
\mathrm{CES}_{n,\alpha}
&=
\mathbb{E}\!\left[
 l_nY_n
 \,\middle|\,
 L\geq\mathrm{VaR}_\alpha
\right].
\label{eq:ces_contribution}
\end{align}
Since
\(L=\sum_n l_nY_n\), these quantities satisfy
\[
\sum_{n=1}^{N}\mathrm{CVaR}_{n,\alpha}=\mathrm{VaR}_\alpha,
\qquad
\sum_{n=1}^{N}\mathrm{CES}_{n,\alpha}=\mathrm{ES}_\alpha.
\]
In smooth models these conditional allocations agree with the familiar Euler
contributions. For a finite discrete portfolio, however, VaR need not be
differentiable with respect to every exposure, so we take the conditional
expectations above as the primary definitions.

When the exact level event is too sparse, it can be replaced by the local band
\begin{equation}
\mathcal{B}_{\alpha,h}
=
\{|L-\mathrm{VaR}_\alpha|\leq h\},
\qquad h\geq0.
\label{eq:local_loss_band}
\end{equation}
The corresponding local VaR contribution is obtained by conditioning on
\(\mathcal{B}_{\alpha,h}\) instead of \(\{L=\mathrm{VaR}_\alpha\}\). For an
integer-valued portfolio with an attainable VaR, one may take \(h=0\).

The importance-sampling construction below therefore depends on the chosen
copula only through the density of \(U\) and the conditional probabilities
\(p_n(u)\); the COS method supplies the smooth tail weight \(q_x(u)\).

\section[Cross-Entropy Importance Sampling in Factor Space]{Cross-Entropy Importance Sampling\\ in Factor Space}
\label{sec:ce_factor_space}

We now choose an importance-sampling proposal for the common state~\(U\).
The idea is straightforward: factor realisations that are often followed by a
large portfolio loss should be sampled more frequently. The conditional tail
probability
\[
q_x(u)=\mathbb{P}(L\geq x\mid U=u)
\]
measures exactly how strongly a factor state~\(u\) is associated with the tail
event. This section uses that quantity to define the ideal, population-level
proposal. Section~\ref{sec:estimating_ce_proposal} explains how the proposal is
estimated in practice.

\subsection{The factor distribution behind a tail loss}
\label{subsec:rare_event_target}

Fix a threshold~\(x\) and write
\begin{equation}
\mathcal{A}_x=\{L\geq x\},
\qquad
\kappa_x=\mathbb{P}(\mathcal{A}_x)>0.
\label{eq:ce_tail_event}
\end{equation}
Bayes' rule gives the density of the common state conditional on this event:
\begin{equation}
\pi_{U,x}^{\star}(u)
=
\frac{f_U(u)q_x(u)}{\kappa_x}
=
f_{U\mid\mathcal{A}_x}(u).
\label{eq:optimal_factor_target_density}
\end{equation}
Thus, a factor state receives more mass whenever it makes the loss event more
likely. Equivalently, for any integrable statistic~\(T\),
\begin{equation}
\mathbb{E}[T(U)\mid\mathcal{A}_x]
=
\frac{\mathbb{E}[T(U)q_x(U)]}{\mathbb{E}[q_x(U)]}.
\label{eq:tail_conditional_moment_identity}
\end{equation}
This identity is the basis of all proposal parameters derived below.

The density in \eqref{eq:optimal_factor_target_density} is the factor marginal
of the usual zero-variance distribution for estimating~\(\kappa_x\). It is
also a zero-variance proposal for the conditional Monte Carlo representation
\(\kappa_x=\mathbb{E}[q_x(U)]\). It is not directly available, however, because
both its normalising constant and the function~\(q_x\) are unknown in closed
form. We therefore approximate it within a tractable family. The same
construction can be used for an exact-level event~\(\{L=x\}\) or a local band
by replacing~\(q_x\) with the corresponding conditional probability.

The word ``optimal'' should be read in this specific sense. The target in
\eqref{eq:optimal_factor_target_density} is natural for conditional Monte
Carlo, and its cross-entropy projection is optimal within the chosen proposal
family. This does not imply that the fitted proposal minimises the variance of
every downstream risk estimator.

\subsection{Cross-entropy as a weighted fit}
\label{subsec:cross_entropy_projection}

Let \(\mathcal{G}=\{g_\theta:\theta\in\Theta\}\) be a parametric family of
factor densities. The cross-entropy parameter is defined by
\begin{equation}
\theta_x^\star
\in
\arg\min_{\theta\in\Theta}
D_{\mathrm{KL}}\!\left(\pi_{U,x}^\star\,\|\,g_\theta\right).
\label{eq:ce_projection_definition}
\end{equation}
Terms that do not depend on~\(\theta\) can be removed from the objective, so
\eqref{eq:ce_projection_definition} is equivalent to
\begin{equation}
\theta_x^\star
\in
\arg\max_{\theta\in\Theta}
\mathbb{E}\!\left[q_x(U)\log g_\theta(U)\right].
\label{eq:ce_weighted_objective}
\end{equation}
The factor~\(\kappa_x^{-1}\) has been omitted because it does not affect the
maximiser. In other words, population cross-entropy calibration is simply a
weighted likelihood fit: factor states are weighted by their conditional
probability of producing a tail loss.

For an exponential family
\[
g_\theta(u)
=
c(u)\exp\{\theta^\top T(u)-A(\theta)\},
\]
the first-order condition for an interior optimum is
\begin{equation}
\mathbb{E}_{g_{\theta_x^\star}}[T(U)]
=
\mathbb{E}[T(U)\mid\mathcal{A}_x].
\label{eq:exponential_family_moment_matching}
\end{equation}
Hence the cross-entropy projection matches the tail-conditional sufficient
statistics. The Gaussian and inverse-Gamma proposals below are both examples
of this rule; see \cite{RubinsteinKroese2004,deBoerEtAl2005} for the general
cross-entropy method.

\subsection{Gaussian proposal for the systematic factors}
\label{subsec:gaussian_ce_projection}

Under both copula specifications, the original systematic factor is
\(Z\sim\mathcal{N}_d(0,I_d)\). We approximate its tail-conditioned marginal by
\(\mathcal{N}_d(\mu,\Sigma)\). Since the Gaussian sufficient statistics are
\(Z\) and \(ZZ^\top\), moment matching gives
\begin{align}
\mu_x^\star
&=
\mathbb{E}[Z\mid\mathcal{A}_x]
=
\frac{\mathbb{E}[Zq_x(U)]}{\mathbb{E}[q_x(U)]},
\label{eq:optimal_gaussian_mean_q}\\
\Sigma_x^\star
&=
\operatorname{Cov}(Z\mid\mathcal{A}_x)
=
\frac{
\mathbb{E}\!\left[
 (Z-\mu_x^\star)(Z-\mu_x^\star)^\top q_x(U)
\right]
}{
\mathbb{E}[q_x(U)]
}.
\label{eq:optimal_gaussian_covariance_q}
\end{align}
Thus the Gaussian proposal is fitted to the global location and spread of the
factor states associated with a large loss. For the Gaussian copula,
\(U=Z\); for the Student~\(t\)-copula, the expectations are over
\(U=(Z,W)\).

If only a mean shift is allowed, the optimal mean is still~\(\mu_x^\star\) and
the covariance is fixed at~\(I_d\). Within a diagonal Gaussian family, the
optimal variances are the diagonal entries of~\(\Sigma_x^\star\). These
restricted proposals can be useful when the pilot sample is too small to fit a
stable full covariance. Finite-sample regularisation is discussed in
Section~\ref{subsec:implementation_covariance_regularisation}.

\subsection{Student \texorpdfstring{$t$}{t}-copula product proposal}
\label{subsec:student_product_proposal}

For the Student~\(t\)-copula, the common state is~\((Z,W)\). Although \(Z\) and
\(W\) are independent under the original model, conditioning on a large loss
generally makes them dependent: a tail loss may be caused by an adverse factor
direction, a large common scale, or both. For tractability, we use the product
family
\begin{equation}
g_{Z,W}(z,w)
=
\phi_d(z;\mu,\Sigma)
 f_W(w;\alpha,\beta),
\label{eq:student_product_proposal}
\end{equation}
where \(f_W\) is the inverse-Gamma density in
\eqref{eq:inverse_gamma_density}. Because the log-density in
\eqref{eq:student_product_proposal} is additive, the cross-entropy objective
separates into a Gaussian fit for~\(Z\) and an inverse-Gamma fit for~\(W\).
The Gaussian parameters are therefore still given by
\eqref{eq:optimal_gaussian_mean_q}--\eqref{eq:optimal_gaussian_covariance_q}.

For the scale component, define the two tail-conditional sufficient-statistic
moments
\begin{align}
m_{\log,x}
&=
\mathbb{E}[\log W\mid\mathcal{A}_x]
=
\frac{\mathbb{E}[(\log W)q_x(Z,W)]}{\mathbb{E}[q_x(Z,W)]},
\label{eq:conditional_log_w_moment}\\
m_{-1,x}
&=
\mathbb{E}[W^{-1}\mid\mathcal{A}_x]
=
\frac{\mathbb{E}[W^{-1}q_x(Z,W)]}{\mathbb{E}[q_x(Z,W)]}.
\label{eq:conditional_inverse_w_moment}
\end{align}

\begin{thm}[Inverse-Gamma cross-entropy projection]
\label{thm:inverse_gamma_ce_projection}
Assume that
\(
\mathbb{E}[|\log W|+W^{-1}\mid\mathcal{A}_x]<\infty
\)
and that \(W\mid\mathcal{A}_x\) is non-degenerate. The inverse-Gamma component
of the product proposal is uniquely determined by
\begin{equation}
\beta_x^\star
=
\frac{\alpha_x^\star}{m_{-1,x}},
\qquad
\psi(\alpha_x^\star)-\log\alpha_x^\star
+m_{\log,x}+\log m_{-1,x}=0,
\label{eq:inverse_gamma_ce_system}
\end{equation}
where \(\psi\) is the digamma function.
\end{thm}

\begin{proof}
For \(\alpha,\beta>0\), the inverse-Gamma log-density is
\[
\log g_{\alpha,\beta}(w)
=
\alpha\log\beta-\log\Gamma(\alpha)
-(\alpha+1)\log w-\frac{\beta}{w}.
\]
Hence the tail-conditional cross-entropy objective is
\[
J_x(\alpha,\beta)
=
\alpha\log\beta-\log\Gamma(\alpha)
-(\alpha+1)m_{\log,x}
-\beta m_{-1,x}.
\]
The assumed integrability ensures that this quantity is finite, and
\(m_{-1,x}>0\).

For fixed \(\alpha>0\),
\[
\frac{\partial J_x}{\partial\beta}
=
\frac{\alpha}{\beta}-m_{-1,x},
\qquad
\frac{\partial^2 J_x}{\partial\beta^2}
=
-\frac{\alpha}{\beta^2}<0.
\]
Therefore the unique maximiser in \(\beta\) is
\[
\beta(\alpha)=\frac{\alpha}{m_{-1,x}}.
\]
Substitution gives the profiled objective
\[
\overline J_x(\alpha)
=
\alpha\log\alpha
-\alpha\log m_{-1,x}
-\log\Gamma(\alpha)
-(\alpha+1)m_{\log,x}
-\alpha .
\]
Let
\[
c_x
=
m_{\log,x}+\log m_{-1,x}.
\]
Then
\[
\overline J_x'(\alpha)
=
\log\alpha-\psi(\alpha)-c_x,
\]
and
\[
\overline J_x''(\alpha)
=
\frac1\alpha-\psi_1(\alpha)<0,
\]
where \(\psi_1\) is the trigamma function. Indeed,
\[
\psi_1(\alpha)
=
\sum_{k=0}^{\infty}\frac{1}{(\alpha+k)^2}
>
\int_0^\infty\frac{dt}{(\alpha+t)^2}
=
\frac1\alpha.
\]
Thus \(\overline J_x\) is strictly concave.

Moreover, strict Jensen inequality gives
\[
\begin{aligned}
c_x
&=
\log\mathbb{E}[W^{-1}\mid\mathcal A_x]
-
\mathbb{E}[\log(W^{-1})\mid\mathcal A_x]
\\
&>0,
\end{aligned}
\]
because \(W\mid\mathcal A_x\) is non-degenerate. Define
\[
h(\alpha)=\log\alpha-\psi(\alpha).
\]
The preceding trigamma inequality shows that \(h\) is strictly decreasing,
while the standard digamma asymptotics give
\[
\lim_{\alpha\downarrow0}h(\alpha)=+\infty,
\qquad
\lim_{\alpha\uparrow\infty}h(\alpha)=0.
\]
Since \(c_x>0\), there is therefore a unique
\(\alpha_x^\star>0\) satisfying
\[
h(\alpha_x^\star)=c_x,
\]
or equivalently,
\[
\psi(\alpha_x^\star)-\log\alpha_x^\star
+m_{\log,x}+\log m_{-1,x}=0.
\]
Because \(\overline J_x\) is strictly concave, this root is its unique global
maximiser. The corresponding scale parameter is uniquely given by
\[
\beta_x^\star
=
\frac{\alpha_x^\star}{m_{-1,x}}.
\]
\end{proof}

The inverse-Gamma fit therefore matches \(\log W\) and~\(W^{-1}\), not the
ordinary mean of~\(W\). The product form in
\eqref{eq:student_product_proposal} is a deliberate approximation: it does not
represent the dependence between \(Z\) and~\(W\) induced by the tail event. A
richer conditional proposal, such as a Gaussian law for \(Z\mid W\), is
possible but is outside the first implementation considered here.

\subsection{Pilot reweighting and likelihood ratios}
\label{subsec:generic_pilot_distribution}

The identities above are written under the original common-state density
\(f_U\), but the pilot sample may come from another density~\(h\). Assume that
\(h(u)>0\) wherever \(f_U(u)q_x(u)>0\), and define
\begin{equation}
R_h(u)=\frac{f_U(u)}{h(u)}.
\label{eq:generic_pilot_likelihood_ratio}
\end{equation}
Then
\begin{equation}
\mathbb{E}[T(U)\mid\mathcal{A}_x]
=
\frac{
\mathbb{E}_h[T(U)R_h(U)q_x(U)]
}{
\mathbb{E}_h[R_h(U)q_x(U)]
}.
\label{eq:generic_pilot_conditional_moment}
\end{equation}
When \(h=f_U\), the weight is simply~\(q_x(U)\). Formula
\eqref{eq:generic_pilot_conditional_moment} also permits iterative
cross-entropy schemes in which the pilot distribution is updated between
stages.

Once a proposal has been fitted, production estimators must be corrected back
to the original measure. If the original Gaussian law is
\(\mathcal{N}_d(0,I_d)\) and the proposal is
\(\mathcal{N}_d(\mu,\Sigma)\), the factor likelihood ratio is
\begin{equation}
R_Z(z;\mu,\Sigma)
=
|\Sigma|^{1/2}
\exp\left\{
\frac12(z-\mu)^\top\Sigma^{-1}(z-\mu)
-\frac12z^\top z
\right\}.
\label{eq:factor_gaussian_likelihood_ratio}
\end{equation}
For a mean shift, \(\Sigma=I_d\), this reduces to
\[
R_Z(z;\mu,I_d)
=
\exp\left(-\mu^\top z+\tfrac12\mu^\top\mu\right).
\]
If the original scale law is
\(\operatorname{InvGamma}(\alpha_0,\beta_0)\) and the proposal uses
\(\operatorname{InvGamma}(\alpha,\beta)\), then
\begin{equation}
R_W(w;\alpha,\beta)
=
\frac{\beta_0^{\alpha_0}\Gamma(\alpha)}
     {\Gamma(\alpha_0)\beta^\alpha}
 w^{\alpha-\alpha_0}
 \exp\left(\frac{\beta-\beta_0}{w}\right).
\label{eq:inverse_gamma_likelihood_ratio}
\end{equation}
Under the product proposal, \(R_U(z,w)=R_Z(z)R_W(w)\). Any change to the
conditional Bernoulli default probabilities contributes a separate likelihood
ratio in the production stage. All ratios are evaluated on the logarithmic
scale in the implementation.

\section{Estimating the Cross-Entropy Proposal}
\label{sec:estimating_ce_proposal}

Section~\ref{sec:ce_factor_space} shows that the cross-entropy proposal is
obtained from moments of the common state under the tail event
\(\mathcal A_x=\{L\geq x\}\). In practice, proposal calibration therefore
comes down to one question: how much weight should each pilot factor draw
receive? CEIS uses the outcome of one simulated portfolio loss, whereas ISCOS
uses an approximation of the conditional probability of a tail loss. Once the
weights have been computed, both methods fit the proposal in exactly the same
way.

Throughout this section, the threshold \(x\) is fixed. In the numerical
experiments it is an attainable estimate of \(\operatorname{VaR}_\alpha\)
obtained from a preliminary Monte Carlo run. The proposal is calibrated to the non-strict event
\(\{L\geq x\}\). The same threshold is then used for the exact-level
and tail estimators in the production stage.

Let \(h\) be the pilot density of the common state and write
\(R_h(u)=f_U(u)/h(u)\). The first implementation takes \(h=f_U\), so that
\(R_h\equiv1\); the more general notation also covers an updated pilot
proposal.

\subsection{One pilot sample, two weights}
\label{subsec:weighted_ce_estimators}

Draw \(U^{(1)},\ldots,U^{(M_0)}\) independently from \(h\). For CEIS, also
generate a conditional default vector and its loss
\(L^{(m)}=\sum_{n=1}^N l_nY_n^{(m)}\). ISCOS does not need this additional
default draw for calibration. The two weights are

\begin{equation}
\omega_m
=
R_h\bigl(U^{(m)}\bigr)
\begin{cases}
\mathbf 1_{\{L^{(m)}\geq x\}}, & \text{CEIS},\\[0.3em]
q_{x,K}\bigl(U^{(m)}\bigr), & \text{ISCOS}.
\end{cases}
\label{eq:ceis_iscos_weight_summary}
\end{equation}

For any nonnegative weights \(\omega_1,\ldots,\omega_{M_0}\), let
\(S_\omega=\sum_m\omega_m\). Provided that \(S_\omega>0\), the Gaussian
proposal is fitted by the weighted mean and covariance

\begin{align}
\widehat\mu_\omega
&=
\frac{1}{S_\omega}
\sum_{m=1}^{M_0}\omega_m Z^{(m)},
\label{eq:weighted_gaussian_mean}
\\
\widehat\Sigma_{\omega,\mathrm{raw}}
&=
\frac{1}{S_\omega}
\sum_{m=1}^{M_0}\omega_m
\bigl(Z^{(m)}-\widehat\mu_\omega\bigr)
\bigl(Z^{(m)}-\widehat\mu_\omega\bigr)^\top.
\label{eq:weighted_gaussian_covariance}
\end{align}

Substituting the first or second line of \eqref{eq:ceis_iscos_weight_summary}
gives the CEIS or ISCOS estimate, respectively. The denominator is the total
weight, rather than an unbiased-sample correction, because the target is the
population covariance in Section~\ref{subsec:gaussian_ce_projection}.

For the Student \(t\)-copula, the same weights are also used to estimate

\[
\widehat m_{\log,\omega}
=
\frac{\sum_m\omega_m\log W^{(m)}}{S_\omega},
\qquad
\widehat m_{-1,\omega}
=
\frac{\sum_m\omega_m(W^{(m)})^{-1}}{S_\omega}.
\]

Inserting these two moments into \eqref{eq:inverse_gamma_ce_system} gives the
fitted inverse-Gamma parameters.

When \(h=f_U\), CEIS assigns weight one only to simulated tail hits. At an
extreme threshold, there may be too few such observations to estimate a
full covariance matrix reliably. ISCOS instead assigns a probability between
zero and one to each factor state. States that are stressful but happen not
to produce a tail loss in one conditional simulation can therefore still
contribute to the fit. Section~\ref{sec:error_analysis} makes this
conditional-Monte-Carlo comparison precise.

\subsection{Computing the smooth COS weights}
\label{subsec:cos_conditional_tail_approximation}

ISCOS requires
\(
q_x(u)=\mathbb P(L\geq x\mid U=u)
\).
Conditional on \(U=u\), the loss characteristic function is available from
\eqref{eq:conditional_characteristic_function}. Choose an interval \([a,b]\)
that contains the loss support in its interior, and let \(K\) be the number
of COS modes, including the zero mode. With
\(
\omega_k=k\pi/(b-a)
\)
and filter coefficients
\(
\sigma_{k,K}=\sigma(k/K)
\), the filtered conditional CDF is approximated by

\begin{align}
F_{L\mid u,K}^{\sigma}(y)
&=
\frac{y-a}{b-a}
+
\frac{2}{\pi}
\sum_{k=1}^{K-1}
\frac{\sigma_{k,K}}{k}
\operatorname{Re}
\left\{
\varphi_{L\mid u}(\omega_k)e^{-\mathrm i\omega_k a}
\right\}
\sin\left(k\pi\frac{y-a}{b-a}\right).
\label{eq:filtered_cos_conditional_cdf}
\end{align}

Because the event is \(L\geq x\), the relevant CDF value is the left limit
\(F_{L\mid u}(x^-)\). We therefore evaluate the expansion at a point
\(y_x<x\) and set

\begin{align}
\widetilde q_{x,K}(u)
&=1-F_{L\mid u,K}^{\sigma}(y_x),
\label{eq:raw_cos_tail_probability}
\\
q_{x,K}(u)
&=\Pi_{[0,1]}\bigl(\widetilde q_{x,K}(u)\bigr).
\label{eq:clipped_cos_tail_probability}
\end{align}

The left-limit convention matters for a discrete loss. If all losses lie on
a lattice with step \(\Delta\) and \(x\) is attainable, we use

\begin{equation}
a=-\frac{\Delta}{2},
\qquad
b=L_{\max}+\frac{\Delta}{2},
\qquad
y_x=x-\frac{\Delta}{2}.
\label{eq:lattice_cos_interval_and_cutoff}
\end{equation}

Then \(F_{L\mid u}(y_x)=\mathbb P(L<x\mid U=u)\), so the COS expansion is not
evaluated at a jump of the CDF. For a non-lattice portfolio, \(y_x\) can be
chosen just below \(x\), with the convention reported alongside the numerical
results.

Clipping removes small Fourier over- and undershoots. Since the true
probability lies in \([0,1]\), it cannot make the pointwise error larger:

\begin{equation}
\bigl|q_{x,K}(u)-q_x(u)\bigr|
\leq
\bigl|\widetilde q_{x,K}(u)-q_x(u)\bigr|.
\label{eq:clipping_nonexpansive}
\end{equation}

The number of modes controls the trade-off between numerical cost and the
accuracy of the smooth weight. Section~\ref{sec:error_analysis} studies how
the finite-COS error propagates to the fitted proposal parameters.

\subsection{Practical calibration}
\label{subsec:implementation_covariance_regularisation}

A preliminary simulation is first used to choose an attainable threshold
\(
\widehat x_\alpha=\inf\{x:\widehat F_{M_0}(x)\geq\alpha\}
\).
The pilot factor draws can then be reused for CEIS and ISCOS, but the final
importance-sampling run should be independent of this calibration sample.

A weighted covariance can be nearly singular when only a small effective
sample contributes. We first symmetrise the raw estimate and, in the baseline
implementation, add a small ridge:

\begin{align}
\widehat\Sigma_{x,\mathrm{sym}}
&=
\frac{1}{2}\left(
\widehat\Sigma_{x,\mathrm{raw}}
+
\widehat\Sigma_{x,\mathrm{raw}}^\top
\right),
\label{eq:covariance_symmetrisation}
\\
\widehat\Sigma_{x,\rho}
&=
\widehat\Sigma_{x,\mathrm{sym}}+\rho I_d,
\qquad \rho>0.
\label{eq:covariance_ridge}
\end{align}

If more regularisation is needed, one may shrink towards a spherical
covariance,

\begin{equation}
\widehat\Sigma_{x,\lambda}
=
(1-\lambda)\widehat\Sigma_{x,\mathrm{sym}}
+
\lambda\tau I_d,
\qquad
\tau=\frac{\operatorname{tr}(\widehat\Sigma_{x,\mathrm{sym}})}{d},
\label{eq:covariance_shrinkage}
\end{equation}

or floor very small eigenvalues. These are finite-sample safeguards rather
than part of the population cross-entropy projection.

A useful summary of the calibration weights is their self-normalised effective
sample size,

\begin{equation}
M_{\mathrm{eff}}
=
\frac{\left(\sum_{m=1}^{M_0}\omega_m\right)^2}
{\sum_{m=1}^{M_0}\omega_m^2}.
\label{eq:calibration_effective_sample_size}
\end{equation}

If \(M_{\mathrm{eff}}\) is comparable to, or smaller than, the factor
dimension, a full covariance fit should be treated cautiously. For ISCOS we
also record the range of the raw COS probabilities, the amount of clipping,
and the average smooth weight. Conditional characteristic functions are
evaluated in batches; identical obligors can be grouped, reducing the direct
\(\mathcal O(M_0NK)\) cost when the portfolio has repeated exposures and
loadings.

The calibration procedure can be summarised as follows.

\begin{enumerate}[label=\arabic*.]
\item Estimate an attainable tail threshold and retain the pilot common states.
\item Compute either the CEIS indicators or the ISCOS conditional-tail
probabilities.
\item Fit the weighted Gaussian moments and, for the Student \(t\)-copula, the
two inverse-Gamma moments.
\item Regularise the covariance when needed and check the weight diagnostics.
\item Use the fitted proposal in an independent production importance-sampling
run.
\end{enumerate}

The two methods therefore differ at only one point: CEIS uses one binary loss
outcome per pilot state, whereas ISCOS uses a smooth estimate of the chance
that the same state produces a tail loss. The next section studies the
statistical gain from this replacement and the approximation error introduced
by the finite COS expansion.

%
%

\section{Error Analysis for ISCOS Calibration}
\label{sec:error_analysis}

The cross-entropy proposal is determined by moments of the common state under
\(\mathcal A_x=\{L\geq x\}\). In practice, these moments are not observed
directly. CEIS estimates them with the binary indicator
\(I_x=\mathbf 1_{\mathcal A_x}\), whereas ISCOS replaces this indicator by a
finite COS approximation \(q_{x,K}(U)\) of the conditional tail probability
\[
q_x(U)=\mathbb P(\mathcal A_x\mid U).
\]
These two substitutions affect the calibration in different ways. Replacing
\(I_x\) by the exact conditional probability \(q_x(U)\) removes the
idiosyncratic default noise that remains after the common state has been drawn.
Replacing \(q_x\) by \(q_{x,K}\), on the other hand, introduces a deterministic
numerical error.

It is useful to keep these effects separate. Let \(\vartheta_x^\star\) denote
the exact population cross-entropy parameters, let
\(\vartheta_{x,K}^\star\) denote the population parameters obtained with the
finite COS weight \(q_{x,K}\), and let
\(\widehat\vartheta_{x,K,M_0}\) be the corresponding estimator from a pilot
sample of size \(M_0\). Then
\begin{equation}
\widehat\vartheta_{x,K,M_0}-\vartheta_x^\star
=
\underbrace{
\bigl(\widehat\vartheta_{x,K,M_0}-\vartheta_{x,K}^\star\bigr)
}_{\text{pilot Monte Carlo error}}
+
\underbrace{
\bigl(\vartheta_{x,K}^\star-\vartheta_x^\star\bigr)
}_{\text{finite-COS error}}.
\label{eq:error_basic_decomposition}
\end{equation}
The first term is statistical and decreases with \(M_0\); the second is
deterministic and decreases with the number of COS modes \(K\). Before
studying this decomposition, we compare CEIS with the ideal estimator that
uses the exact conditional probability.

Throughout this section, the loss threshold \(x\) is fixed. Thus, all results
refer to the calibration of a proposal for a prescribed event
\(\mathcal A_x\). The additional effect of replacing \(x\) by a preliminary
empirical quantile is discussed in
Section~\ref{subsec:error_scope_and_threshold}.

\subsection{Exact conditional weighting and Rao--Blackwellisation}
\label{subsec:error_raoblackwell}

The proposal parameters are ratios of unnormalised tail moments. It is
therefore convenient to compare the underlying raw moments first, before the
normalisation is applied. Let \(T(U)\in\mathbb R^r\) collect the common-state
statistics required by the chosen proposal family. For the Gaussian
component, one may take
\begin{equation}
T_{\mathrm G}(U)
=
\left(
1,
Z^\top,
\operatorname{vech}(ZZ^\top)^\top
\right)^\top.
\label{eq:error_gaussian_statistic_vector}
\end{equation}
The first coordinate estimates the tail probability, the next \(d\)
coordinates estimate the numerator of the tail-conditional mean, and the
remaining coordinates estimate the numerator of the second moment. For the
Student~\(t\)-copula proposal, the inverse-Gamma sufficient statistics are
appended:
\begin{equation}
T_t(U)
=
\left(
T_{\mathrm G}(U)^\top,
\log W,
W^{-1}
\right)^\top.
\label{eq:error_t_statistic_vector}
\end{equation}

Let \(h\) be the pilot density and
\(R_h(u)=f_U(u)/h(u)\) the corresponding likelihood ratio. Define the
one-observation raw-moment vectors
\begin{equation}
G_x^{\mathrm{CEIS}}
:=R_h(U)I_xT(U),
\qquad
G_x^{\mathrm{RB}}
:=R_h(U)q_x(U)T(U).
\label{eq:error_ceis_rb_raw_vectors}
\end{equation}
Here \(\mathrm{RB}\) denotes the ideal Rao--Blackwellised estimator.
Conditional on \(U\), the factor \(R_h(U)T(U)\) is fixed, and therefore
\begin{equation}
\mathbb E\!\left[G_x^{\mathrm{CEIS}}\mid U\right]
=R_h(U)T(U)\mathbb E[I_x\mid U]
=G_x^{\mathrm{RB}}.
\label{eq:error_rb_conditional_expectation}
\end{equation}
Consequently, the two vectors estimate exactly the same population moments:
\begin{equation}
\mathbb E_h\!\left[G_x^{\mathrm{CEIS}}\right]
=
\mathbb E_h\!\left[G_x^{\mathrm{RB}}\right]
=
\mathbb E_{f_U}\!\left[q_x(U)T(U)\right].
\label{eq:error_equal_raw_moment_means}
\end{equation}
The difference lies only in their sampling covariance.

\begin{thm}[Exact Rao--Blackwell covariance reduction]
\label{thm:error_exact_covariance_reduction}
Assume that
\[
\mathbb E_h\!\left[R_h(U)^2\|T(U)\|_2^2\right]<\infty.
\]
Then
\begin{align}
&\operatorname{Cov}_h\!\left(G_x^{\mathrm{CEIS}}\right)
-
\operatorname{Cov}_h\!\left(G_x^{\mathrm{RB}}\right)
\notag\\
&\qquad=
\mathbb E_h\!\left[
R_h(U)^2q_x(U)\bigl(1-q_x(U)\bigr)
T(U)T(U)^\top
\right]
\succeq0.
\label{eq:error_exact_covariance_difference}
\end{align}
Equivalently, for every \(a\in\mathbb R^r\),
\begin{equation}
\operatorname{Var}_h\!\left(a^\top G_x^{\mathrm{RB}}\right)
\leq
\operatorname{Var}_h\!\left(a^\top G_x^{\mathrm{CEIS}}\right).
\label{eq:error_scalar_variance_order}
\end{equation}
\end{thm}

\begin{proof}
Conditional on \(U\), the only remaining randomness in
\(G_x^{\mathrm{CEIS}}\) is the Bernoulli variable \(I_x\). Since
\[
\mathbb E[I_x\mid U]=q_x(U),
\qquad
\operatorname{Var}(I_x\mid U)
=q_x(U)\bigl(1-q_x(U)\bigr),
\]
we have \eqref{eq:error_rb_conditional_expectation} and
\[
\operatorname{Cov}\!\left(G_x^{\mathrm{CEIS}}\mid U\right)
=
R_h(U)^2q_x(U)\bigl(1-q_x(U)\bigr)T(U)T(U)^\top.
\]
The law of total covariance gives
\begin{align*}
\operatorname{Cov}_h\!\left(G_x^{\mathrm{CEIS}}\right)
&=
\operatorname{Cov}_h\!\left(
\mathbb E[G_x^{\mathrm{CEIS}}\mid U]
\right)
+
\mathbb E_h\!\left[
\operatorname{Cov}(G_x^{\mathrm{CEIS}}\mid U)
\right]\\
&=
\operatorname{Cov}_h\!\left(G_x^{\mathrm{RB}}\right)
+
\mathbb E_h\!\left[
R_h(U)^2q_x(U)\bigl(1-q_x(U)\bigr)T(U)T(U)^\top
\right],
\end{align*}
which proves \eqref{eq:error_exact_covariance_difference}. The last matrix is
positive semidefinite because it is an expectation of positive-semidefinite
rank-one matrices. Equation~\eqref{eq:error_scalar_variance_order} follows by
pre- and post-multiplying by \(a^\top\) and \(a\).
\end{proof}

The result has a direct componentwise interpretation. When the original pilot
law is used, so that \(R_h\equiv1\), every scalar statistic \(T_\ell(U)\)
satisfies
\begin{equation}
\operatorname{Var}\!\left(T_\ell(U)I_x\right)
-
\operatorname{Var}\!\left(T_\ell(U)q_x(U)\right)
=
\mathbb E\!\left[
T_\ell(U)^2q_x(U)\bigl(1-q_x(U)\bigr)
\right].
\label{eq:error_componentwise_variance_gain}
\end{equation}
Thus, the denominator, every first moment, and every second moment used in the
Gaussian fit all benefit from the same conditional-variance removal. The gain
is strict in a direction \(a\) whenever
\(a^\top R_h(U)T(U)\neq0\) with positive probability on a region where
\(0<q_x(U)<1\). If \(q_x(U)\in\{0,1\}\) almost surely, the common state already
determines whether the event occurs, so there is no conditional default noise
to remove.

The exact order in Theorem~\ref{thm:error_exact_covariance_reduction} concerns
the unnormalised moments. The fitted CE parameters are nonlinear functions of
these moments. To make this connection explicit, write the Gaussian raw moment
vector as
\begin{equation}
m_x=(\kappa_x,b_x,\operatorname{vech}(C_x)),
\qquad
b_x=\mathbb E_h[R_hq_xZ],
\qquad
C_x=\mathbb E_h[R_hq_xZZ^\top],
\label{eq:error_gaussian_raw_moments}
\end{equation}
where \(\kappa_x=\mathbb E_h[R_hq_x]>0\). The parameter map
\(\mathcal G\) is
\begin{equation}
\mu_x^\star=\frac{b_x}{\kappa_x},
\qquad
\Sigma_x^\star
=\frac{C_x}{\kappa_x}-\mu_x^\star(\mu_x^\star)^\top.
\label{eq:error_gaussian_parameter_map}
\end{equation}
For perturbations \((\dot\kappa,\dot b,\dot C)\), its differential is
\begin{align}
D\mu_x^\star
&=
\frac{\dot b-\mu_x^\star\dot\kappa}{\kappa_x},
\label{eq:error_gaussian_mean_differential}\\
D\Sigma_x^\star
&=
\frac{\dot C}{\kappa_x}
-
\frac{C_x\dot\kappa}{\kappa_x^2}
-
(D\mu_x^\star)(\mu_x^\star)^\top
-
\mu_x^\star(D\mu_x^\star)^\top.
\label{eq:error_gaussian_covariance_differential}
\end{align}
Hence the map is continuously differentiable whenever \(\kappa_x>0\).

\begin{coro}[Asymptotic covariance of the fitted proposal]
\label{coro:error_parameter_covariance_reduction}
Assume the conditions of
Theorem~\ref{thm:error_exact_covariance_reduction}, let \(\kappa_x>0\), and
suppose that the parameter map \(\mathcal G\) is continuously differentiable
at the common raw-moment vector \(m_x\). Define
\begin{equation}
\widehat m_x^{\,j}
=
\frac1{M_0}\sum_{m=1}^{M_0}G_{x,m}^{j},
\qquad
j\in\{\mathrm{CEIS},\mathrm{RB}\}.
\label{eq:error_raw_moment_sample_averages}
\end{equation}
Then
\begin{equation}
\sqrt{M_0}
\left(
\mathcal G(\widehat m_x^{\,j})-\mathcal G(m_x)
\right)
\xrightarrow{\mathrm d}
\mathcal N(0,\mathcal V_j),
\label{eq:error_parameter_clt}
\end{equation}
where
\begin{equation}
\mathcal V_j
=
D\mathcal G(m_x)
\operatorname{Cov}_h(G_x^j)
D\mathcal G(m_x)^\top.
\label{eq:error_parameter_asymptotic_covariance}
\end{equation}
Moreover,
\begin{equation}
\mathcal V_{\mathrm{RB}}
\preceq
\mathcal V_{\mathrm{CEIS}}.
\label{eq:error_parameter_asymptotic_covariance_order}
\end{equation}
\end{coro}

\begin{proof}
The multivariate central limit theorem applies to the raw-moment averages in
\eqref{eq:error_raw_moment_sample_averages}. The delta method then yields
\eqref{eq:error_parameter_clt}. Finally,
\begin{align*}
\mathcal V_{\mathrm{CEIS}}-\mathcal V_{\mathrm{RB}}
&=
D\mathcal G(m_x)
\left[
\operatorname{Cov}_h(G_x^{\mathrm{CEIS}})
-
\operatorname{Cov}_h(G_x^{\mathrm{RB}})
\right]
D\mathcal G(m_x)^\top,
\end{align*}
which is positive semidefinite by
Theorem~\ref{thm:error_exact_covariance_reduction}.
\end{proof}

For the Student~\(t\)-copula, the same conclusion applies to the Gaussian
parameters and to the inverse-Gamma parameters. The latter are smooth
functions of the two sufficient-statistic moments \(m_{\log,x}\) and
\(m_{-1,x}\) because the derivative of the inverse-Gamma shape equation is
non-zero at the population solution; this point is made explicit in
Corollary~\ref{coro:error_inverse_gamma_parameter_rate} below.

Corollary~\ref{coro:error_parameter_covariance_reduction} is a first-order
asymptotic statement for the normalised proposal parameters. It does not give
an exact finite-sample variance order for the ratio estimators
\(\widehat\mu_x\) and \(\widehat\Sigma_x\). It also does not yet describe
practical ISCOS, because practical ISCOS uses \(q_{x,K}\), not the exact
conditional probability \(q_x\).

\subsection{A uniform COS error for the conditional tail probability}
\label{subsec:error_uniform_cos}

We next study the deterministic difference between \(q_{x,K}\) and \(q_x\).
The required bound must hold uniformly over the common state: the subsequent
CE moments integrate over all possible values of \(U\), so a pointwise-in-
\(u\) convergence statement would not by itself control the moment ratios.
The finite support of the conditional portfolio loss makes the desired
uniformity available.

Let
\begin{equation}
\mathcal S_L=\{s_1,\ldots,s_J\}\subset(a,b)
\label{eq:error_finite_loss_support}
\end{equation}
be the distinct attainable losses. This set depends only on the deterministic
losses-at-default and not on the common state. Conditional on \(U=u\), write
\begin{equation}
\pi_j(u)=\mathbb P(L=s_j\mid U=u),
\qquad
\pi_j(u)\geq0,
\qquad
\sum_{j=1}^J\pi_j(u)=1.
\label{eq:error_conditional_atom_probabilities}
\end{equation}
The conditional CDF is therefore the convex combination
\begin{equation}
F_{L\mid u}(y)
=
\sum_{j=1}^J\pi_j(u)\mathbf 1_{\{s_j\leq y\}}.
\label{eq:error_conditional_cdf_atom_mixture}
\end{equation}

Let \(y_x<x\) be the continuity point used in Section~4 to represent the
left limit, so that
\(q_x(u)=1-F_{L\mid u}(y_x)\). For a unit point mass at \(s_j\), denote by
\(H_{j,K}^{\sigma}(y_x)\) the filtered COS approximation to its CDF at
\(y_x\). The discrete-COS results in \cite{ShenFangLiu2024} give, for the
chosen filter, interval, and evaluation point, an atomwise estimate of the
form
\begin{equation}
\left|
H_{j,K}^{\sigma}(y_x)-\mathbf 1_{\{s_j\leq y_x\}}
\right|
\leq
C_{j,x,\sigma}K^{-p},
\qquad j=1,\ldots,J,
\label{eq:error_atomwise_discrete_cos_rate}
\end{equation}
where \(p>0\) is the proved convergence order of the filtered CDF
approximation.

\begin{thm}[Uniform conditional COS error]
\label{thm:error_uniform_conditional_cos_rate}
Suppose that \(\mathcal S_L\) is finite, lies strictly inside \((a,b)\), and
\(y_x\notin\mathcal S_L\). If
\eqref{eq:error_atomwise_discrete_cos_rate} holds, then there is a finite
constant \(C_{x,\sigma}\), independent of \(u\), such that
\begin{equation}
\sup_u
\left|
\widetilde q_{x,K}(u)-q_x(u)
\right|
\leq
C_{x,\sigma}K^{-p}.
\label{eq:error_uniform_raw_tail_rate}
\end{equation}
The clipped COS probability in
\eqref{eq:clipped_cos_tail_probability} satisfies the same order:
\begin{equation}
\epsilon_{x,K}
:=
\sup_u\left|q_{x,K}(u)-q_x(u)\right|
=O(K^{-p}).
\label{eq:error_uniform_clipped_tail_rate}
\end{equation}
\end{thm}

\begin{proof}
The filtered COS reconstruction is linear in the characteristic function and
therefore linear in the atom probabilities. Hence
\begin{equation}
F_{L\mid u,K}^{\sigma}(y_x)
=
\sum_{j=1}^J\pi_j(u)H_{j,K}^{\sigma}(y_x).
\label{eq:error_cos_cdf_atom_mixture}
\end{equation}
Combining \eqref{eq:error_conditional_cdf_atom_mixture} and
\eqref{eq:error_cos_cdf_atom_mixture} gives
\begin{align*}
\left|
F_{L\mid u,K}^{\sigma}(y_x)-F_{L\mid u}(y_x)
\right|
&\leq
\sum_{j=1}^J\pi_j(u)
\left|
H_{j,K}^{\sigma}(y_x)-\mathbf 1_{\{s_j\leq y_x\}}
\right|\\
&\leq
\left(\max_{1\leq j\leq J}C_{j,x,\sigma}\right)
K^{-p}
\sum_{j=1}^J\pi_j(u)\\
&=
C_{x,\sigma}K^{-p}.
\end{align*}
The constant is independent of \(u\) because the conditional atom
probabilities form a convex combination. Since
\(\widetilde q_{x,K}=1-F_{L\mid u,K}^{\sigma}(y_x)\), the same bound holds for
the raw tail approximation. Finally, projection onto \([0,1]\) is
non-expansive relative to any point in \([0,1]\), so
\eqref{eq:clipping_nonexpansive} gives the same bound for \(q_{x,K}\).
\end{proof}

The assumption \(y_x\notin\mathcal S_L\) is essential for approximating the
non-strict event \(\{L\geq x\}\). At a jump of a discrete CDF, a symmetric
Fourier reconstruction naturally approaches a midpoint value rather than the
left limit required by \(1-F_{L\mid u}(x^-)\). For a lattice portfolio, the
half-step convention in
\eqref{eq:lattice_cos_interval_and_cutoff} places \(y_x\) strictly between two
attainable losses and therefore avoids this ambiguity.

Two features of Theorem~\ref{thm:error_uniform_conditional_cos_rate} are worth
emphasising. First, the common state changes the probabilities
\(\pi_j(u)\), but not the support points \(s_j\); this is why the bound is
uniform in both the Gaussian and Student~\(t\)-copula models. Second, the
symbol \(p\) denotes the convergence order actually established for the
selected filtered CDF approximation. It should not be identified
mechanically with a parameter in the name of a filter. For example, if the
available theorem for a formal order-\(r\) filter gives only
\(O(K^{1-r})\), then the subsequent results should be read with
\(p=r-1\). A sharper value of \(p\) should be used only when it is proved
for the selected filter and parameterisation. Filter-specific endpoint terms
and the position of \(y_x\) relative to the support and the COS interval
affect the constant and, in some cases, the available order.

\subsection{Propagation of the COS error to the CE parameters}
\label{subsec:error_parameter_propagation}

The finite-COS proposal is defined through normalised weighted moments, so the
next step is to propagate the uniform error in
\eqref{eq:error_uniform_clipped_tail_rate} through these ratios. We first give
a simple identity that will be used repeatedly.

\begin{lemma}[Perturbation identity for a weighted moment]
\label{lemma:error_weighted_ratio_identity}
Let \(q\) and \(q_K\) be integrable weights with
\(\kappa=\mathbb E[q(U)]>0\) and
\(\kappa_K=\mathbb E[q_K(U)]>0\). For an integrable vector- or matrix-valued
function \(g\), define
\[
\bar g
=
\frac{\mathbb E[g(U)q(U)]}{\kappa},
\qquad
\bar g_K
=
\frac{\mathbb E[g(U)q_K(U)]}{\kappa_K}.
\]
With \(e_K=q_K-q\),
\begin{equation}
\bar g_K-\bar g
=
\frac{
\mathbb E\!\left[(g(U)-\bar g)e_K(U)\right]
}{\kappa_K}.
\label{eq:error_weighted_ratio_identity}
\end{equation}
Consequently, if \(\|e_K\|_\infty\leq\epsilon_K\), then
\begin{equation}
\|\bar g_K-\bar g\|
\leq
\frac{\epsilon_K}{\kappa_K}
\mathbb E\!\left[\|g(U)-\bar g\|\right].
\label{eq:error_weighted_ratio_bound}
\end{equation}
\end{lemma}

\begin{proof}
Using \(q_K=q+e_K\) and
\(\mathbb E[gq]=\kappa\bar g\),
\begin{align*}
\kappa_K(\bar g_K-\bar g)
&=
\mathbb E[gq_K]-\kappa_K\bar g\\
&=
\mathbb E[gq]+\mathbb E[ge_K]
-(\kappa+\mathbb E[e_K])\bar g\\
&=
\mathbb E[(g-\bar g)e_K].
\end{align*}
The norm bound follows immediately.
\end{proof}

For later use, define the finite-\(K\) normaliser and Gaussian
cross-entropy parameters by
\begin{align}
\kappa_{x,K}
&:=
\mathbb{E}_h\!\left[
R_h(U)q_{x,K}(U)
\right]
=
\mathbb{E}_{f_U}\!\left[
q_{x,K}(U)
\right],
\label{eq:cos_population_tail_normalizer}
\\
\mu_{x,K}^{\star}
&:=
\frac{
\mathbb{E}_h\!\left[
Z R_h(U)q_{x,K}(U)
\right]
}{
\kappa_{x,K}
},
\label{eq:cos_population_gaussian_mean}
\\
\Sigma_{x,K}^{\star}
&:=
\frac{
\mathbb{E}_h\!\left[
\bigl(Z-\mu_{x,K}^{\star}\bigr)
\bigl(Z-\mu_{x,K}^{\star}\bigr)^\top
R_h(U)q_{x,K}(U)
\right]
}{
\kappa_{x,K}
}.
\label{eq:cos_population_gaussian_covariance}
\end{align}
Their exact counterparts are obtained by replacing \(q_{x,K}\) with
\(q_x\).

\begin{thm}[Propagation to the Gaussian CE parameters]
\label{thm:error_gaussian_ce_parameter_rate}
Assume that \(\kappa_x=\mathbb P(L\geq x)>0\),
\(\mathbb E[\|Z\|_2^2]<\infty\), and
\eqref{eq:error_uniform_clipped_tail_rate} holds. Then
\begin{equation}
|\kappa_{x,K}-\kappa_x|
\leq
\epsilon_{x,K}.
\label{eq:error_kappa_explicit_bound}
\end{equation}
For all sufficiently large \(K\), so that
\(\epsilon_{x,K}\leq\kappa_x/2\),
\begin{equation}
\|\mu_{x,K}^\star-\mu_x^\star\|_2
\leq
\frac{2\epsilon_{x,K}}{\kappa_x}
\mathbb E\!\left[\|Z-\mu_x^\star\|_2\right],
\label{eq:error_mu_explicit_bound}
\end{equation}
and
\begin{align}
\|\Sigma_{x,K}^\star-\Sigma_x^\star\|_{\mathrm F}
&\leq
\frac{2\epsilon_{x,K}}{\kappa_x}
\mathbb E\!\left[
\left\|
(Z-\mu_x^\star)(Z-\mu_x^\star)^\top
-\Sigma_x^\star
\right\|_{\mathrm F}
\right]
\notag\\
&\quad+
\left(
\frac{2\epsilon_{x,K}}{\kappa_x}
\mathbb E\!\left[\|Z-\mu_x^\star\|_2\right]
\right)^2.
\label{eq:error_sigma_explicit_bound}
\end{align}
In particular,
\begin{equation}
|\kappa_{x,K}-\kappa_x|=O(K^{-p}),
\qquad
\|\mu_{x,K}^\star-\mu_x^\star\|_2=O(K^{-p}),
\qquad
\|\Sigma_{x,K}^\star-\Sigma_x^\star\|_{\mathrm F}=O(K^{-p}).
\label{eq:error_gaussian_parameter_rates}
\end{equation}
\end{thm}

\begin{proof}
Let \(e_{x,K}=q_{x,K}-q_x\). The normalising constants satisfy
\[
|\kappa_{x,K}-\kappa_x|
=
|\mathbb E[e_{x,K}(U)]|
\leq
\epsilon_{x,K},
\]
which proves \eqref{eq:error_kappa_explicit_bound}. Hence
\(\kappa_{x,K}\geq\kappa_x/2\) for sufficiently large \(K\).

Apply Lemma~\ref{lemma:error_weighted_ratio_identity} with
\(g(U)=Z\), \(q=q_x\), and \(q_K=q_{x,K}\). Since the exact weighted mean is
\(\mu_x^\star\),
\begin{equation}
\mu_{x,K}^\star-\mu_x^\star
=
\frac{
\mathbb E\!\left[(Z-\mu_x^\star)e_{x,K}(U)\right]
}{\kappa_{x,K}}.
\label{eq:error_mean_exact_identity}
\end{equation}
Taking norms, using the uniform error bound, and replacing
\(1/\kappa_{x,K}\) by \(2/\kappa_x\) gives
\eqref{eq:error_mu_explicit_bound}.

For the covariance, set
\[
A_x(Z)
=
(Z-\mu_x^\star)(Z-\mu_x^\star)^\top,
\qquad
\delta_{x,K}
=
\mu_{x,K}^\star-\mu_x^\star.
\]
Because the \(q_{x,K}\)-weighted mean of
\(Z-\mu_x^\star\) is \(\delta_{x,K}\), expanding around
\(\mu_x^\star\) yields
\begin{equation}
\Sigma_{x,K}^\star
=
\frac{\mathbb E[A_x(Z)q_{x,K}(U)]}{\kappa_{x,K}}
-
\delta_{x,K}\delta_{x,K}^\top.
\label{eq:error_covariance_expansion_around_exact_mean}
\end{equation}
The exact \(q_x\)-weighted mean of \(A_x(Z)\) is
\(\Sigma_x^\star\). Applying
Lemma~\ref{lemma:error_weighted_ratio_identity} to \(A_x\) and using
\eqref{eq:error_covariance_expansion_around_exact_mean} gives
\begin{equation}
\Sigma_{x,K}^\star-\Sigma_x^\star
=
\frac{
\mathbb E\!\left[
\bigl(A_x(Z)-\Sigma_x^\star\bigr)e_{x,K}(U)
\right]
}{\kappa_{x,K}}
-
\delta_{x,K}\delta_{x,K}^\top.
\label{eq:error_covariance_exact_identity}
\end{equation}
Taking Frobenius norms and using
\(\|\delta\delta^\top\|_{\mathrm F}=\|\delta\|_2^2\) proves
\eqref{eq:error_sigma_explicit_bound}. The moment assumption makes the
constants finite, and
\(\epsilon_{x,K}=O(K^{-p})\) gives
\eqref{eq:error_gaussian_parameter_rates}.
\end{proof}

The explicit bounds show what is hidden by the order notation. The algebraic
order is inherited without loss, but the constants contain
\(1/\kappa_x\). Therefore, the theorem is a fixed-threshold result; it is not
a uniform statement as \(x\) moves indefinitely deeper into the tail. For a
rarer event, a given absolute error in \(q_{x,K}\) can produce a larger error
in the normalised CE moments.

The same perturbation identity controls the scale component of the
Student~\(t\)-copula proposal. Define
\begin{equation}
m_{\log,x,K}
=
\frac{\mathbb E[(\log W)q_{x,K}(Z,W)]}{\kappa_{x,K}},
\qquad
m_{-1,x,K}
=
\frac{\mathbb E[W^{-1}q_{x,K}(Z,W)]}{\kappa_{x,K}}.
\label{eq:error_finite_k_inverse_gamma_moments}
\end{equation}

\begin{coro}[Student \(t\)-copula scale parameters]
\label{coro:error_inverse_gamma_parameter_rate}
Assume the conditions of
Theorem~\ref{thm:inverse_gamma_ce_projection},
\[
\mathbb E[|\log W|+W^{-1}]<\infty,
\]
and \eqref{eq:error_uniform_clipped_tail_rate}. Let
\((\alpha_{x,K}^\star,\beta_{x,K}^\star)\) be the inverse-Gamma CE
parameters obtained from the finite-COS moments in
\eqref{eq:error_finite_k_inverse_gamma_moments}. Then
\begin{equation}
|m_{\log,x,K}-m_{\log,x}|
+
|m_{-1,x,K}-m_{-1,x}|
=O(K^{-p}),
\label{eq:error_inverse_gamma_moment_rate}
\end{equation}
and
\begin{equation}
|\alpha_{x,K}^\star-\alpha_x^\star|
+
|\beta_{x,K}^\star-\beta_x^\star|
=O(K^{-p}).
\label{eq:error_inverse_gamma_parameter_rate}
\end{equation}
\end{coro}

\begin{proof}
Applying Lemma~\ref{lemma:error_weighted_ratio_identity} to
\(g(W)=\log W\) and \(g(W)=W^{-1}\) gives
\eqref{eq:error_inverse_gamma_moment_rate}. The inverse-Gamma shape parameter
is defined by
\begin{equation}
H(\alpha,m_{\log},m_{-1})
:=
\psi(\alpha)-\log\alpha+m_{\log}+\log m_{-1}=0.
\label{eq:error_inverse_gamma_implicit_equation}
\end{equation}
At the exact solution,
\begin{equation}
\frac{\partial H}{\partial\alpha}
=
\psi_1(\alpha_x^\star)-\frac1{\alpha_x^\star}>0,
\label{eq:error_inverse_gamma_nonzero_derivative}
\end{equation}
where \(\psi_1\) is the trigamma function. The implicit function theorem
therefore makes \(\alpha\) locally continuously differentiable, and hence
locally Lipschitz, in \((m_{\log},m_{-1})\). Finally,
\(\beta=\alpha/m_{-1}\) is smooth because \(m_{-1,x}>0\). This proves
\eqref{eq:error_inverse_gamma_parameter_rate}.
\end{proof}

\subsection{Finite-COS covariance and the combined calibration error}
\label{subsec:error_combined}

The exact covariance order in
Theorem~\ref{thm:error_exact_covariance_reduction} does not hold literally
for a finite COS expansion, because \(q_{x,K}(U)\) is not exactly
\(\mathbb E[I_x\mid U]\). Nevertheless, the finite-\(K\) raw-moment
covariance approaches the Rao--Blackwell covariance at the same rate as the
conditional-probability approximation.

Define
\begin{equation}
G_{x,K}^{\mathrm{ISCOS}}
:=
R_h(U)q_{x,K}(U)T(U),
\qquad
V_{x,K}
:=
\operatorname{Cov}_h(G_{x,K}^{\mathrm{ISCOS}}),
\qquad
V_x^{\mathrm{RB}}
:=
\operatorname{Cov}_h(G_x^{\mathrm{RB}}).
\label{eq:error_finite_k_raw_covariances}
\end{equation}

\begin{prop}[Finite-\(K\) perturbation of the Rao--Blackwell covariance]
\label{prop:error_finite_k_covariance_perturbation}
Assume \eqref{eq:error_uniform_clipped_tail_rate} and
\[
\mathbb E_h\!\left[R_h(U)^2\|T(U)\|_2^2\right]<\infty.
\]
Then
\begin{equation}
\|V_{x,K}-V_x^{\mathrm{RB}}\|_2
=O(K^{-p}).
\label{eq:error_finite_k_covariance_rate}
\end{equation}
If the parameter map \(\mathcal G\) is twice continuously differentiable in
a neighbourhood of the target raw moments, the corresponding finite-\(K\)
asymptotic covariance of the fitted parameters is also an
\(O(K^{-p})\) perturbation of the exact Rao--Blackwell asymptotic covariance.
\end{prop}

\begin{proof}
Write
\begin{equation}
D_{x,K}
:=
G_{x,K}^{\mathrm{ISCOS}}-G_x^{\mathrm{RB}}
=
R_h(U)\bigl(q_{x,K}(U)-q_x(U)\bigr)T(U).
\label{eq:error_finite_k_raw_difference}
\end{equation}
The uniform COS bound implies
\begin{equation}
\|D_{x,K}\|_{L^2(h)}
\leq
\epsilon_{x,K}
\left(
\mathbb E_h[R_h(U)^2\|T(U)\|_2^2]
\right)^{1/2}
=O(K^{-p}).
\label{eq:error_finite_k_l2_difference}
\end{equation}
Now expand
\[
\operatorname{Cov}(G_x^{\mathrm{RB}}+D_{x,K})
-
\operatorname{Cov}(G_x^{\mathrm{RB}})
=
\operatorname{Cov}(G_x^{\mathrm{RB}},D_{x,K})
+
\operatorname{Cov}(D_{x,K},G_x^{\mathrm{RB}})
+
\operatorname{Cov}(D_{x,K}).
\]
Cauchy--Schwarz bounds the two cross-covariance terms by a constant times
\(\|D_{x,K}\|_{L^2(h)}\), while the final term is of order
\(\|D_{x,K}\|_{L^2(h)}^2\). This proves
\eqref{eq:error_finite_k_covariance_rate}. The parameter statement follows
by combining this bound with the \(O(K^{-p})\) perturbation of the population
raw moments and the continuity of the Jacobian of \(\mathcal G\).
\end{proof}

Proposition~\ref{prop:error_finite_k_covariance_perturbation} gives a norm
convergence result, not a finite-\(K\) Loewner order. At a particular
resolution, the deterministic approximation error can move the covariance in
either direction. What is guaranteed is that, as \(K\) increases, the
finite-COS covariance approaches the lower Rao--Blackwell covariance.

We now add the pilot sampling error. For clarity, define the finite-\(K\) raw
moment average
\begin{equation}
\widehat m_{x,K}
=
\frac1{M_0}
\sum_{m=1}^{M_0}
R_h(U^{(m)})q_{x,K}(U^{(m)})T(U^{(m)}),
\label{eq:error_finite_k_raw_moment_average}
\end{equation}
with population mean \(m_{x,K}\). When the ISCOS weights in
\eqref{eq:ceis_iscos_weight_summary} are used, denote the estimators in
\eqref{eq:weighted_gaussian_mean}--
\eqref{eq:weighted_gaussian_covariance} by
\[
\widehat{\mu}_{x,K}^{\mathrm{ISCOS}}
:=
\widehat{\mu}_{\omega},
\qquad
\widehat{\Sigma}_{x,K,\mathrm{raw}}^{\mathrm{ISCOS}}
:=
\widehat{\Sigma}_{\omega,\mathrm{raw}}.
\]
These estimators are components of
\(\mathcal G(\widehat m_{x,K})\), while
\((\mu_{x,K}^{\star},\Sigma_{x,K}^{\star})\) are components of
\(\mathcal G(m_{x,K})\).

\begin{thm}[Combined ISCOS calibration error]
\label{thm:error_combined_iscos_rate}
Suppose that \(\kappa_x>0\),
\eqref{eq:error_uniform_clipped_tail_rate} holds, and
\begin{equation}
\mathbb E_h\!\left[
R_h(U)^2\bigl(1+\|Z\|_2^4\bigr)
\right]<\infty.
\label{eq:error_pilot_moment_condition}
\end{equation}
Then, for the fixed threshold \(x\),
\begin{equation}
\left\|
\widehat\mu_{x,K}^{\mathrm{ISCOS}}-\mu_x^\star
\right\|_2
=
O_{\mathbb P}(M_0^{-1/2})+O(K^{-p}),
\label{eq:error_combined_mu_rate}
\end{equation}
and
\begin{equation}
\left\|
\widehat\Sigma_{x,K,\mathrm{raw}}^{\mathrm{ISCOS}}
-\Sigma_x^\star
\right\|_{\mathrm F}
=
O_{\mathbb P}(M_0^{-1/2})+O(K^{-p}).
\label{eq:error_combined_sigma_rate}
\end{equation}
The same decomposition holds for the Student~\(t\)-copula Gaussian and
inverse-Gamma parameters under the additional moment conditions used in
Corollary~\ref{coro:error_inverse_gamma_parameter_rate}.
\end{thm}

\begin{proof}
Because the clipped COS weights satisfy \(0\leq q_{x,K}\leq1\), the moment
condition \eqref{eq:error_pilot_moment_condition} gives a multivariate central
limit theorem for the sample normaliser, first moment, and second moment. By
Theorem~\ref{thm:error_uniform_conditional_cos_rate},
\(\kappa_{x,K}\geq\kappa_x/2\) for all sufficiently large \(K\), so the
ratio map is smooth in a neighbourhood of the finite-\(K\) population
moments. The delta method therefore gives
\begin{equation}
\widehat\mu_{x,K}^{\mathrm{ISCOS}}-\mu_{x,K}^\star
=O_{\mathbb P}(M_0^{-1/2}),
\qquad
\widehat\Sigma_{x,K,\mathrm{raw}}^{\mathrm{ISCOS}}
-\Sigma_{x,K}^\star
=O_{\mathbb P}(M_0^{-1/2}).
\label{eq:error_finite_k_sampling_rates}
\end{equation}
Using the decomposition
\begin{align*}
\widehat\mu_{x,K}^{\mathrm{ISCOS}}-\mu_x^\star
&=
\bigl(
\widehat\mu_{x,K}^{\mathrm{ISCOS}}-\mu_{x,K}^\star
\bigr)
+
\bigl(
\mu_{x,K}^\star-\mu_x^\star
\bigr),\\
\widehat\Sigma_{x,K,\mathrm{raw}}^{\mathrm{ISCOS}}-\Sigma_x^\star
&=
\bigl(
\widehat\Sigma_{x,K,\mathrm{raw}}^{\mathrm{ISCOS}}
-\Sigma_{x,K}^\star
\bigr)
+
\bigl(
\Sigma_{x,K}^\star-\Sigma_x^\star
\bigr),
\end{align*}
and applying
Theorem~\ref{thm:error_gaussian_ce_parameter_rate} proves
\eqref{eq:error_combined_mu_rate}--\eqref{eq:error_combined_sigma_rate}.
The Student~\(t\)-copula conclusion follows in the same way from the smooth
parameter maps and Corollary~\ref{coro:error_inverse_gamma_parameter_rate}.
\end{proof}

Theorem~\ref{thm:error_combined_iscos_rate} separates the two tuning
parameters clearly. For a fixed \(K\), increasing \(M_0\) removes the pilot
sampling error but leaves convergence to the finite-COS target rather than to
the exact CE target. To make the COS error negligible on the root-\(M_0\)
scale, one needs
\begin{equation}
K^{-p}=o(M_0^{-1/2}).
\label{eq:error_root_m_negligible_cos_condition}
\end{equation}
Balancing the two displayed orders gives the heuristic choice
\begin{equation}
K\asymp M_0^{1/(2p)}.
\label{eq:error_balanced_k_choice}
\end{equation}
This balances deterministic and statistical error only. It is not necessarily
work-optimal, because the direct calibration cost grows approximately linearly
in both \(M_0\) and \(K\).

Under the usual additional uniform-integrability conditions for the ratio
estimators, the componentwise mean-squared error corresponding to
\eqref{eq:error_combined_mu_rate}--\eqref{eq:error_combined_sigma_rate} has
the form
\begin{equation}
\operatorname{MSE}
=
O(M_0^{-1})+O(K^{-2p}).
\label{eq:error_componentwise_mse_rate}
\end{equation}
The qualification is relevant because a ratio estimator can have unusually
large finite-sample moments if its random denominator is allowed to approach
zero. In implementation, the total calibration weight and the effective
sample size should therefore be monitored alongside the asymptotic rate.

Covariance regularisation contributes an additional deterministic term. For
the ridge estimator in \eqref{eq:covariance_ridge}, and fixed dimension \(d\),
\begin{equation}
\left\|
\widehat\Sigma_{x,K,\rho}^{\mathrm{ISCOS}}-\Sigma_x^\star
\right\|_{\mathrm F}
=
O_{\mathbb P}(M_0^{-1/2})+O(K^{-p})+O(\rho).
\label{eq:error_regularised_ridge_rate}
\end{equation}
For the shrinkage estimator in \eqref{eq:covariance_shrinkage}, the analogous
additional term is \(O(\lambda)\). Consistency for the unregularised
population covariance requires \(\rho\to0\) and \(\lambda\to0\). If the
regularisation parameters are held fixed, the estimator instead converges to
the corresponding regularised target, which may nevertheless be preferable
for finite-sample importance-sampling stability.

\subsection{Threshold estimation and scope of the result}
\label{subsec:error_scope_and_threshold}

The preceding results condition on a fixed threshold \(x\). In the numerical
implementation, \(x\) is obtained from a preliminary empirical VaR. For a
discrete loss distribution, this additional step is qualitatively different
from the smooth moment perturbations considered above.

Let
\(x_\alpha=\inf\{y:F_L(y)\geq\alpha\}\). If the confidence level lies
strictly inside the jump at \(x_\alpha\), namely
\begin{equation}
F_L(x_\alpha^-)<\alpha<F_L(x_\alpha),
\label{eq:error_quantile_gap_condition}
\end{equation}
then the empirical quantile is locally stable. More precisely, for an
independent preliminary sample of size \(M_{\mathrm{pre}}\), define
\[
\delta_-
=
\alpha-F_L(x_\alpha^-)>0,
\qquad
\delta_+
=
F_L(x_\alpha)-\alpha>0.
\]
Hoeffding's inequality gives the elementary bound
\begin{equation}
\mathbb P(\widehat x_\alpha\neq x_\alpha)
\leq
\exp(-2M_{\mathrm{pre}}\delta_-^2)
+
\exp(-2M_{\mathrm{pre}}\delta_+^2).
\label{eq:error_empirical_quantile_stability_bound}
\end{equation}
Indeed, selecting a loss below \(x_\alpha\) requires the empirical CDF just
below \(x_\alpha\) to cross \(\alpha\), while selecting a larger loss requires
the empirical CDF at \(x_\alpha\) to fall below \(\alpha\). Under
\eqref{eq:error_quantile_gap_condition}, both events are large deviations.

If \(\alpha\) coincides with one side of a CDF jump, one of these gaps
vanishes and the empirical quantile can move to an adjacent atom under a
small CDF perturbation. In that case, the threshold error cannot in general be
absorbed into the smooth
\(O_{\mathbb P}(M_0^{-1/2})+O(K^{-p})\) expansion. It should be reported
separately, controlled by a larger independent preliminary sample, or removed
from the calibration comparison by fixing a common threshold across methods.

Finally, the covariance comparison in this section concerns proposal
calibration, not every downstream importance-sampling estimator. A more
accurate or lower-variance estimate of the CE moments need not produce a
smaller variance for every final VaR, expected-shortfall, or obligor-level
ratio estimator. The production variance also depends on the nonlinear
likelihood ratio, the conditional Bernoulli twisting, the selected
regularisation, and the particular final integrand. Moreover, the CE proposal
is the Kullback--Leibler projection of the tail-conditioned factor law within
the chosen family; it is not the unrestricted variance-minimising proposal for
every possible downstream quantity.

The main conclusions can therefore be summarised as
\begin{equation}
\boxed{
\operatorname{Cov}_h(G_x^{\mathrm{RB}})
\preceq
\operatorname{Cov}_h(G_x^{\mathrm{CEIS}})
}
\label{eq:error_summary_covariance_order}
\end{equation}
for exact conditional weighting,
\begin{equation}
\boxed{
\epsilon_{x,K}=O(K^{-p})
\quad\Longrightarrow\quad
\|\mu_{x,K}^\star-\mu_x^\star\|_2
+
\|\Sigma_{x,K}^\star-\Sigma_x^\star\|_{\mathrm F}
=O(K^{-p})
}
\label{eq:error_summary_parameter_rate}
\end{equation}
for the deterministic COS approximation, and
\begin{equation}
\boxed{
\widehat\vartheta_{x,K,M_0}-\vartheta_x^\star
=
O_{\mathbb P}(M_0^{-1/2})+O(K^{-p})
}
\label{eq:error_summary_combined_rate}
\end{equation}
for the combined calibration error, componentwise for the Gaussian proposal
parameters and, under the stated moment conditions, for the Student~\(t\)
scale parameters.

%

\section{Numerical Experiments: Gaussian Copula}
\label{sec:numerical_gaussian}

Section~\ref{sec:error_analysis} separates the sampling error of proposal
calibration from the deterministic error of the finite COS approximation. This
section studies both effects on an eleven-factor Gaussian-copula benchmark. The
first experiment fixes the pilot factors and compares filtered COS weights with
exact conditional tail probabilities obtained by block convolution. The second
uses a common pilot, common production budgets, and common random-number streams
to compare three practical factor proposals: the paper-aligned Glasserman mean
shift \(\mathcal N(\mu_x^{\mathrm{GL}},I_{11})\), indicator-weighted CEIS, and
filtered ISCOS with exactly 32 modes. The Hessian covariance extension of the
Glasserman construction is not used in this comparison.

We retain the historical terminology used throughout the manuscript:
\begin{equation}
\mathrm{CVaR}_{k,\alpha}
=
\mathbb E[l_kY_k\mid L=\mathrm{VaR}_\alpha],
\qquad
\mathrm{CES}_{k,\alpha}
=
\mathbb E[l_kY_k\mid L\geq\mathrm{VaR}_\alpha].
\label{eq:num_gaussian_contribution_definitions}
\end{equation}
For a discrete loss, the threshold-tail mean
\(\mathbb E[L\mid L\geq\mathrm{VaR}_\alpha]\) need not equal the
quantile-integral definition of expected shortfall because the VaR atom may
require partial inclusion.

All intervals reported below are conditional, nominal, marginal pointwise
95\% intervals from the production ratio estimators. They do not propagate the
uncertainty of the preliminary VaR, proposal calibration, or the finite COS
approximation. They are not intervals for pairwise method differences, and
interval overlap is not used as a significance test.

\subsection{Benchmark portfolio and matched experimental design}
\label{subsec:gaussian_benchmark_setup}

We use Example~4 of Glasserman~\cite{Glasserman2005}. The portfolio contains
100 obligors divided into ten homogeneous blocks of ten. Every obligor has
unconditional default probability \(0.01\). The first of eleven independent
standard Gaussian factors is market-wide, while factor \(g+1\) is specific to
block \(g\). The loss-at-default in block \(g\) is \(c_g\), with
\begin{equation}
(c_1,\ldots,c_{10})=(1,1,4,4,9,9,16,16,25,25).
\label{eq:num_block_exposures}
\end{equation}
Thus each exposure level \(1,4,9,16,25\) occurs for twenty obligors.

The implementation uses the convention that positive factor realisations are
adverse. For an obligor in block \(g\),
\begin{equation}
p_g(z)
=
\Phi\!\left(
\frac{\Phi^{-1}(0.01)+0.3z_1+0.8z_{g+1}}{\sqrt{0.27}}
\right).
\label{eq:num_example4_conditional_pd}
\end{equation}
This is equivalent to the lower-tail convention of
Section~\ref{sec:credit_portfolio_models} after reversing the signs of the
non-zero loadings. Conditional on \(Z=z\), the block default counts are
independent and
\begin{equation}
D_g\mid Z=z\sim\operatorname{Binomial}\bigl(10,p_g(z)\bigr),
\qquad
L\mid Z=z=\sum_{g=1}^{10}c_gD_g.
\label{eq:num_block_model}
\end{equation}
The loss lies on the integer lattice and satisfies \(0\leq L\leq1100\).

A preliminary Monte Carlo sample gives the attainable \(99.9\%\) threshold
\(x=250\), consistent with the benchmark in~\cite{Glasserman2005}. The main
settings are collected in Table~\ref{tab:num_reproducibility_parameters}.
The mode count includes the zero mode throughout.

\begin{table}[htbp]
\centering
\caption{Experimental design for the Gaussian-copula experiments.}
\label{tab:num_reproducibility_parameters}
\small
\begin{tabular}{@{}p{0.31\textwidth}p{0.63\textwidth}@{}}
\toprule
Quantity & Specification \\
\midrule
Confidence level and threshold
& \(\alpha=0.999\), \(x=\widehat{\mathrm{VaR}}_\alpha=250\) \\
Shared pilot sample
& \(M_0=\num{250000}\), seed 42 \\
Production budget per method
& \(M=\num{250000}\) for \(L=250\) and \(M=\num{250000}\) for \(L\geq250\) \\
Production seeds
& exact-level event: 420; threshold-tail event: 421 \\
Random-number coupling
& the same base Gaussian draws and Bernoulli uniforms are used across methods within each event; the transformed factors and defaults still differ by proposal \\
Exact-convolution validation
& \(K\in\{16,32,64,128,256,512,1024\}\), EXP-4 filter \\
Matched-budget ISCOS
& exactly \(K=32\) modes, \(k=0,\ldots,31\), with \(\sigma_{k,32}=\exp[-8(k/32)^4]\) \\
COS interval and lattice cutoff
& \([a,b]=[-0.5,1100.5]\), \(y_x=249.5\), loss step \(\Delta=1\) \\
Implementation details
& batches of \(\num{5000}\); ten homogeneous blocks; ridge \(\rho=10^{-8}\); raw COS probabilities projected onto \([0,1]\) \\
\bottomrule
\end{tabular}
\end{table}

The three methods differ only in the Gaussian factor proposal used in the
production stage. The benchmark proposals are
\begin{align}
g_{\mathrm{GL}}(z)
&=
\phi_{11}(z;\mu_x^{\mathrm{GL}},I_{11}),
\label{eq:num_glasserman_mean_only_proposal}
\\
g_{\mathrm{CEIS}}(z)
&=
\phi_{11}(z;\widehat\mu_x^{\mathrm{CEIS}},
                 \widehat\Sigma_x^{\mathrm{CEIS}}),
\qquad
\omega_m^{\mathrm{CEIS}}=\mathbf 1_{\{L^{(m)}\geq250\}},
\label{eq:num_ceis_proposal}
\\
g_{\mathrm{ISCOS}}(z)
&=
\phi_{11}(z;\widehat\mu_{x,32}^{\mathrm{ISCOS}},
                 \widehat\Sigma_{x,32}^{\mathrm{ISCOS}}),
\qquad
\omega_m^{\mathrm{ISCOS}}=q_{x,32}(Z^{(m)}).
\label{eq:num_iscos32_proposal}
\end{align}
The Glasserman mean is obtained by the saddle-point optimisation used in the
paper's Algorithms~6.1 and~6.2, while its covariance remains the identity.
CEIS and ISCOS use the same weighted Gaussian fit and the same ridge; only
their calibration weights differ.

All three proposals are followed by the same conditional Bernoulli twist,
\begin{equation}
p_n^{(\theta)}(z)
=
\frac{p_n(z)e^{\theta l_n}}
     {1+p_n(z)(e^{\theta l_n}-1)},
\qquad
\psi(\theta,z)
=
\sum_{n=1}^{N}
\log\!\left[1+p_n(z)(e^{\theta l_n}-1)\right].
\label{eq:num_conditional_twist}
\end{equation}
For \(\{L=250\}\), \(\theta\) solves
\(\sum_n l_np_n^{(\theta)}(z)=250\). For \(\{L\geq250\}\), the
implementation uses \(\theta^+(z)=\max\{\theta(z),0\}\), following
\cite{GlassermanLi2005,Glasserman2005}. If \(R_Z\) denotes the Gaussian
factor likelihood ratio, the complete production likelihood ratio is
\[
R_Z(z)\exp\{-\theta L+\psi(\theta,z)\}.
\]
Thus the credit model, conditional twist, event definitions, and final ratio
estimators are identical across methods.

\subsection{Accuracy of the filtered COS calibration}
\label{subsec:num_exact_cos_benchmark}

The block structure permits an exact conditional benchmark. For each pilot
factor draw \(Z^{(m)}\), finite-state convolution of the ten binomial block
counts gives
\begin{equation}
q_x^{\mathrm{ex}}(Z^{(m)})
=
\mathbb P\!\left(
\sum_{g=1}^{10}c_gD_g\geq250
\,\middle|\,
Z^{(m)}
\right).
\label{eq:num_exact_conditional_tail}
\end{equation}
The dynamic program retains states below the threshold and accumulates mass as
soon as the threshold is crossed. This avoids subtracting two nearly equal
probabilities when the conditional tail is very small.

On the common finite pilot, the exact-weight Gaussian fit is
\begin{align}
\widehat\mu_x^{\mathrm{ex}}
&=
\frac{\sum_m q_x^{\mathrm{ex}}(Z^{(m)})Z^{(m)}}
     {\sum_m q_x^{\mathrm{ex}}(Z^{(m)})},
\label{eq:num_exact_pilot_mu}
\\
\widehat\Sigma_x^{\mathrm{ex}}
&=
\frac{\sum_m q_x^{\mathrm{ex}}(Z^{(m)})
 (Z^{(m)}-\widehat\mu_x^{\mathrm{ex}})
 (Z^{(m)}-\widehat\mu_x^{\mathrm{ex}})^\top}
 {\sum_m q_x^{\mathrm{ex}}(Z^{(m)})}
 +10^{-8}I_{11}.
\label{eq:num_exact_pilot_sigma}
\end{align}
These quantities are exact only with respect to the conditional loss
calculation; they remain finite-pilot estimates of the population CE moments.
For the filtered COS fit, define
\begin{equation}
e_\mu(K)
=
\frac{\|\widehat\mu_{x,K}^{\mathrm{ISCOS}}-
        \widehat\mu_x^{\mathrm{ex}}\|_2}
     {\|\widehat\mu_x^{\mathrm{ex}}\|_2},
\qquad
e_\Sigma(K)
=
\frac{\|\widehat\Sigma_{x,K}^{\mathrm{ISCOS}}-
        \widehat\Sigma_x^{\mathrm{ex}}\|_{\mathrm F}}
     {\|\widehat\Sigma_x^{\mathrm{ex}}\|_{\mathrm F}}.
\label{eq:num_proposal_relative_errors}
\end{equation}
The exact conditional weights have sample mean
\(\num{1.19926e-3}\) and calibration ESS \(\num{587.60}\). In
Table~\ref{tab:num_exact_cos_benchmark}, ``spurious mass'' is the share of
total clipped COS weight carried by pilot states with
\(q_x^{\mathrm{ex}}(Z^{(m)})\leq10^{-12}\).

\begin{table}[htbp]
\centering
\caption{EXP-4 COS calibration against exact conditional block convolution
on the same \(\num{250000}\) pilot factors.}
\label{tab:num_exact_cos_benchmark}
\small
\begin{tabular}{@{}rrrrrr@{}}
\toprule
\(K\)
& \(M_0^{-1}\sum_m q_{x,K}^{(m)}\)
& mean \(\lvert q_{x,K}-q_x^{\mathrm{ex}}\rvert\)
& spurious mass
& \(e_\mu(K)\)
& \(e_\Sigma(K)\) \\
\midrule
16   & \num{1.1699e-2} & \num{1.071e-2} & 86.7408\% & 94.59\% & 58.67\% \\
32   & \num{1.1070e-3} & \num{2.082e-4} &  2.9546\% &  5.67\% &  8.95\% \\
64   & \num{1.0369e-3} & \num{1.801e-4} &  0.2448\% &  6.44\% &  7.82\% \\
128  & \num{1.0542e-3} & \num{1.568e-4} &  0.0012\% &  5.72\% &  6.72\% \\
256  & \num{1.0764e-3} & \num{1.354e-4} &  0.0762\% &  4.78\% &  5.81\% \\
512  & \num{1.0985e-3} & \num{1.073e-4} &  0.0276\% &  3.88\% &  4.71\% \\
1024 & \num{1.1301e-3} & \num{7.125e-5} &  0.0008\% &  2.53\% &  3.07\% \\
\bottomrule
\end{tabular}
\end{table}

\begin{figure}[htbp]
\centering
\includegraphics[width=0.78\textwidth]
{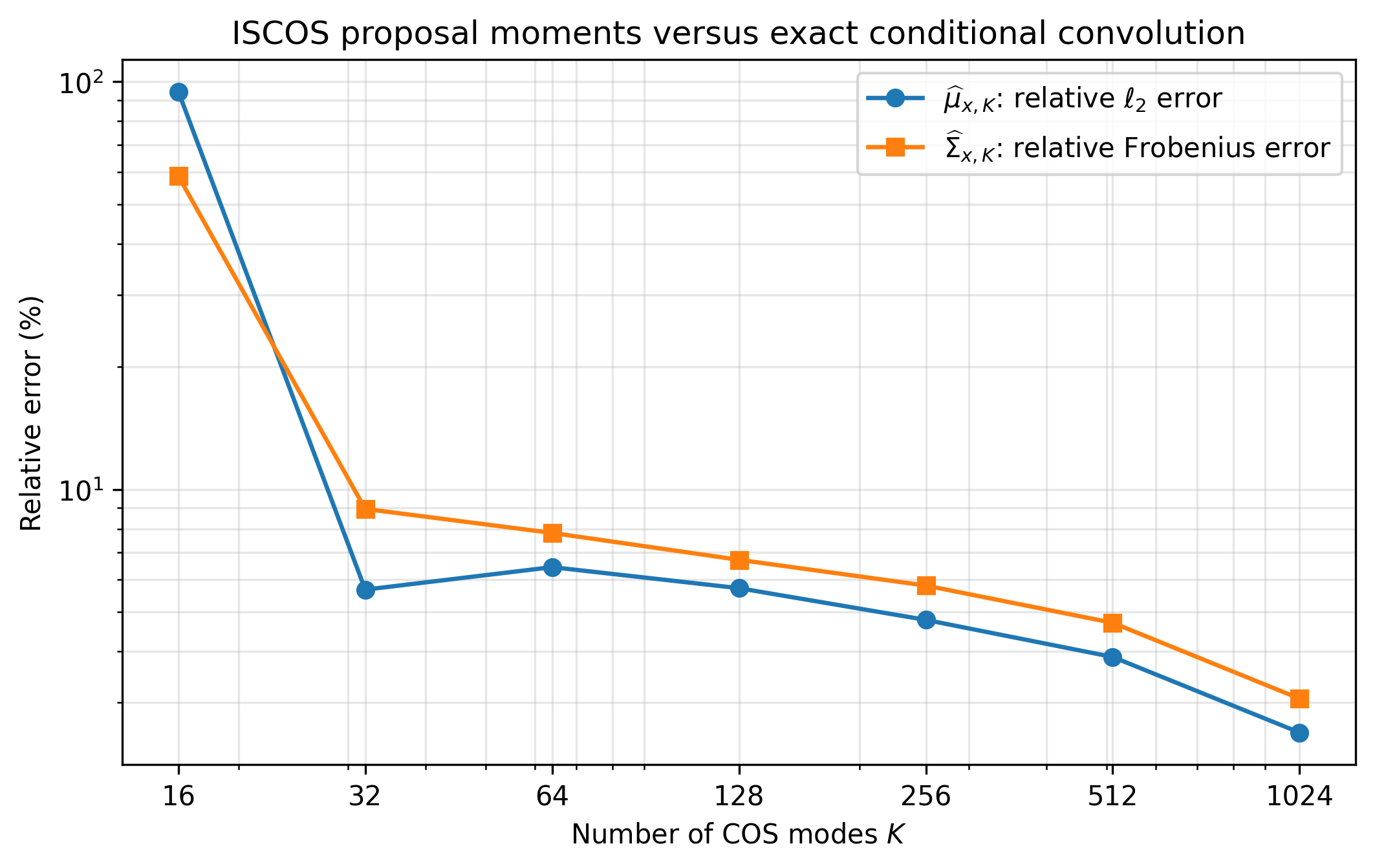}
\caption{Relative errors of the EXP-4 ISCOS proposal moments against exact
conditional-convolution moments on the same pilot sample.}
\label{fig:num_exact_proposal_convergence}
\end{figure}

The 16-mode approximation is inadequate: most of its fitted weight is assigned
to factor states whose exact conditional tail probability is negligible. The
largest improvement occurs between \(K=16\) and \(K=32\). Thereafter the mean
absolute weight error and the proposal-moment errors decrease overall, although
not monotonically at every adjacent mode count. At \(K=1024\), the relative
mean and covariance errors are \(2.53\%\) and \(3.07\%\), respectively. The
finite-pilot, clipping, and near-zero-weight effects visible in the table make
this a convergence diagnostic rather than an estimate of the asymptotic
exponent in Section~\ref{sec:error_analysis}.

The matched-budget ISCOS run uses exactly the filtered \(K=32\) specification
in the second row of Table~\ref{tab:num_exact_cos_benchmark}. Before clipping,
its COS probabilities have sample mean \(\num{1.06296e-3}\), minimum
\(\num{-4.29299e-2}\), and maximum \(\num{1.004806}\). The fractions below
zero and above one are \(8.7728\%\) and \(0.0020\%\), respectively. After
projection onto \([0,1]\), the weight mean is \(\num{1.10700e-3}\) and the
calibration ESS is \(\num{601.8}\). These clipping fractions are useful
numerical diagnostics but are not pointwise error measures; the exact
conditional comparison is the relevant accuracy check.

\subsection{Proposal-calibration diagnostics}
\label{subsec:num_proposal_method_comparison}

For a fitted Gaussian proposal \((\widehat\mu,\widehat\Sigma)\), define its
relative distance from the exact-weight finite-pilot fit by
\begin{equation}
d_\mu
=
\frac{\|\widehat\mu-\widehat\mu_x^{\mathrm{ex}}\|_2}
     {\|\widehat\mu_x^{\mathrm{ex}}\|_2},
\qquad
d_\Sigma
=
\frac{\|\widehat\Sigma-\widehat\Sigma_x^{\mathrm{ex}}\|_{\mathrm F}}
     {\|\widehat\Sigma_x^{\mathrm{ex}}\|_{\mathrm F}}.
\label{eq:num_distances_to_exact_ce_fit}
\end{equation}
For the Glasserman proposal these are descriptive distances, because its
mean-only mode approximation is not designed to match the CE covariance.

We also audit the Gaussian factor likelihood ratio. If
\(f=\mathcal N(0,I_d)\), \(g=\mathcal N(\mu,\Sigma)\), and \(R_Z=f/g\), then
\begin{equation}
\mathbb E_g[R_Z^2]<\infty
\quad\Longleftrightarrow\quad
2I_d-\Sigma^{-1}\succ0
\quad\Longleftrightarrow\quad
\lambda_{\min}(\Sigma)>\tfrac12.
\label{eq:num_gaussian_lr_second_moment_condition}
\end{equation}
We report the margin
\begin{equation}
m_{\mathrm{LR}}(\Sigma)
=
\lambda_{\min}(2I_d-\Sigma^{-1}),
\label{eq:num_lr_margin}
\end{equation}
so a positive value means that the untruncated Gaussian factor ratio passes
this second-moment check.

\begin{table}[htbp]
\centering
\caption{Proposal-calibration diagnostics under the common pilot budget. The
exact conditional CE row uses block-convolution weights on the same finite
pilot sample.}
\label{tab:num_proposal_calibration_comparison}
\small
\resizebox{\textwidth}{!}{%
\begin{tabular}{@{}lrrrrrr@{}}
\toprule
Method
& mean weight
& calibration ESS
& \(d_\mu\)
& \(d_\Sigma\)
& \(\lambda_{\min}(\widehat\Sigma)\)
& \(m_{\mathrm{LR}}(\widehat\Sigma)\) \\
\midrule
Exact conditional CE
& \num{1.1993e-3} & \num{587.6} & 0.00\% & 0.00\% & 0.5050 & 0.0197 \\
Glasserman \(\mathcal N(\mu^{\mathrm{GL}},I)\)
& \multicolumn{1}{c}{--} & \multicolumn{1}{c}{--} & 36.26\% & 59.66\% & 1.0000 & 1.0000 \\
CEIS
& \num{1.1760e-3} & \num{294.0} & 7.04\% & 10.99\% & 0.4726 & -0.1158 \\
ISCOS-32 (EXP-4)
& \num{1.1070e-3} & \num{601.8} & 5.67\% & 8.95\% & 0.7296 & 0.6294 \\
\bottomrule
\end{tabular}%
}
\end{table}

The CEIS pilot contains 294 realised tail observations, so its binary-weight
ESS is exactly 294. The smooth ISCOS weights raise the calibration ESS to
\(601.8\), an increase of \(104.7\%\) under the same pilot budget. In this run,
filtered ISCOS-32 is also closer to the exact-weight fit than CEIS in both
reported norms: \(d_\mu\) falls from \(7.04\%\) to \(5.67\%\), and
\(d_\Sigma\) from \(10.99\%\) to \(8.95\%\). The exact-weight covariance is
itself close to the likelihood-ratio boundary, with
\(m_{\mathrm{LR}}=0.0197\). Sampling noise moves the CEIS covariance just
across that boundary, whereas the filtered ISCOS covariance remains more
diffuse and has a positive margin of \(0.6294\).

The Glasserman mean-only proposal automatically satisfies the Gaussian
second-moment condition because its covariance is \(I_{11}\). Its fitted mean
shifts the market factor and the two industry factors attached to the
largest-loss blocks most strongly, in line with the portfolio's exposure
profile. Its distances from the exact CE moments should not be interpreted as
failure of its own objective: it is a local mode-based approximation with a
fixed covariance, whereas CEIS and ISCOS are global moment fits.

\subsection{Risk contributions and nominal pointwise intervals}
\label{subsec:num_gaussian_risk_contributions}

Let \(A\) denote either \(\{L=250\}\) or \(\{L\geq250\}\), and let
\(\Lambda^{(m)}\) be the complete factor-and-default likelihood ratio. The
obligor-level ratio estimator and the standard error used for the displayed
bands are
\begin{align}
\widehat r_k(A)
&=
\frac{\sum_{m=1}^{M}\Lambda^{(m)}l_kY_k^{(m)}
      \mathbf 1_{A^{(m)}}}
     {\sum_{m=1}^{M}\Lambda^{(m)}\mathbf 1_{A^{(m)}}},
\label{eq:num_ratio_contribution_estimator}
\\
\widehat{\operatorname{se}}\bigl(\widehat r_k(A)\bigr)
&=
\frac{
\left[\sum_{m=1}^{M}
\left\{\Lambda^{(m)}\mathbf 1_{A^{(m)}}
\bigl(l_kY_k^{(m)}-\widehat r_k(A)\bigr)\right\}^{2}\right]^{1/2}}
{\sum_{m=1}^{M}\Lambda^{(m)}\mathbf 1_{A^{(m)}}}.
\label{eq:num_ratio_standard_error}
\end{align}
The reported interval is
\(\widehat r_k\pm1.96\widehat{\operatorname{se}}(\widehat r_k)\).
The event-weight ESS is
\[
M_{\mathrm{eff},A}
=
\frac{\bigl(\sum_m\Lambda^{(m)}\mathbf 1_{A^{(m)}}\bigr)^2}
     {\sum_m\bigl(\Lambda^{(m)}\mathbf 1_{A^{(m)}}\bigr)^2}.
\]

\begin{table}[htbp]
\centering
\caption{Downstream importance-sampling diagnostics. Half-lengths are
averages over the 100 obligors. All intervals are nominal and subject to the
second-moment qualifications discussed in the text.}
\label{tab:num_downstream_comparison}
\small
\resizebox{\textwidth}{!}{%
\begin{tabular}{@{}lrrrrrrr@{}}
\toprule
Method
& \shortstack{CVaR\\hit rate}
& \shortstack{CVaR\\ESS}
& \shortstack{mean CVaR\\half-length}
& \shortstack{CES\\hit rate}
& \shortstack{CES\\ESS}
& \shortstack{mean CES\\half-length}
& \shortstack{tail-loss\\mean} \\
\midrule
Glasserman \(\mathcal N(\mu^{\mathrm{GL}},I)\)
& 12.6040\% & \num{402.5}  & 0.2667 & 75.2372\% & \num{1291.6}  & 0.2560 & 280.003 \\
CEIS
& 13.8720\% & \num{5337.8} & 0.07623 & 70.3120\% & \num{29354.6} & 0.05068 & 280.470 \\
ISCOS-32 (EXP-4)
& 13.1556\% & \num{9404.8} & 0.05993 & 70.2516\% & \num{45413.0} & 0.03916 & 280.508 \\
\bottomrule
\end{tabular}%
}
\end{table}

The estimated exact-event probabilities are
\(\num{2.77015e-4}\), \(\num{2.72541e-4}\), and
\(\num{2.74853e-4}\) for Glasserman, CEIS, and ISCOS, respectively. The
corresponding tail-event estimates are \(\num{1.10660e-3}\),
\(\num{1.15660e-3}\), and \(\num{1.15360e-3}\). All three CVaR allocation
vectors add to 250 up to numerical rounding, while each CES vector adds to its
estimated threshold-tail mean. For CEIS and ISCOS, these means are 280.470 and
280.508, a difference below \(0.02\%\). Their contribution vectors are also
close: the ISCOS--CEIS relative \(\ell_2\) differences are \(1.52\%\) for
CVaR and \(0.39\%\) for CES. In the \(l_k=25\) group, an obligor contributes
about 12.35 to CVaR and 11.47 to CES under both methods.

Raw proposal hit rates do not rank the methods by efficiency. Glasserman
produces many event hits, but its event-weight ESS is only \(402.5\) for CVaR
and \(1291.6\) for CES because a small number of likelihood weights dominate.
Relative to CEIS, ISCOS raises the exact-event ESS by \(76.2\%\) and the
tail-event ESS by \(54.7\%\), despite slightly lower hit rates. Its average
nominal half-length is \(21.4\%\) smaller for CVaR and \(22.7\%\) smaller for
CES. The componentwise CVaR ordering is mixed---ISCOS is narrower for 50 of
100 obligors---whereas it is narrower for 99 of 100 CES components. The
reported average reductions therefore describe this matched run and should not
be read as a uniform componentwise ordering.

\begin{figure}[p]
\centering
\begin{subfigure}[t]{0.49\textwidth}
\centering
\includegraphics[width=\textwidth]
{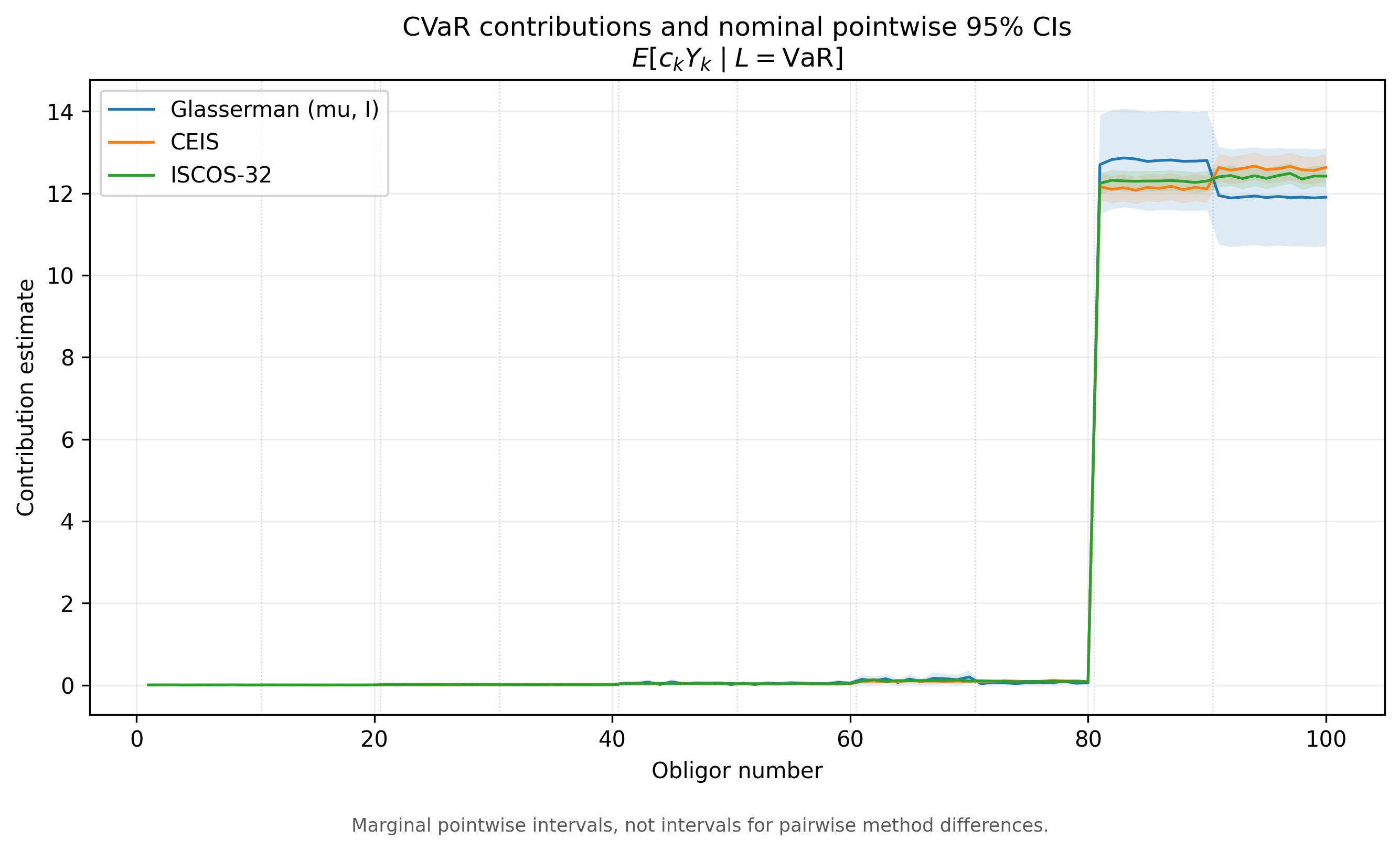}
\caption{CVaR contributions conditional on \(L=250\).}
\end{subfigure}
\hfill
\begin{subfigure}[t]{0.49\textwidth}
\centering
\includegraphics[width=\textwidth]
{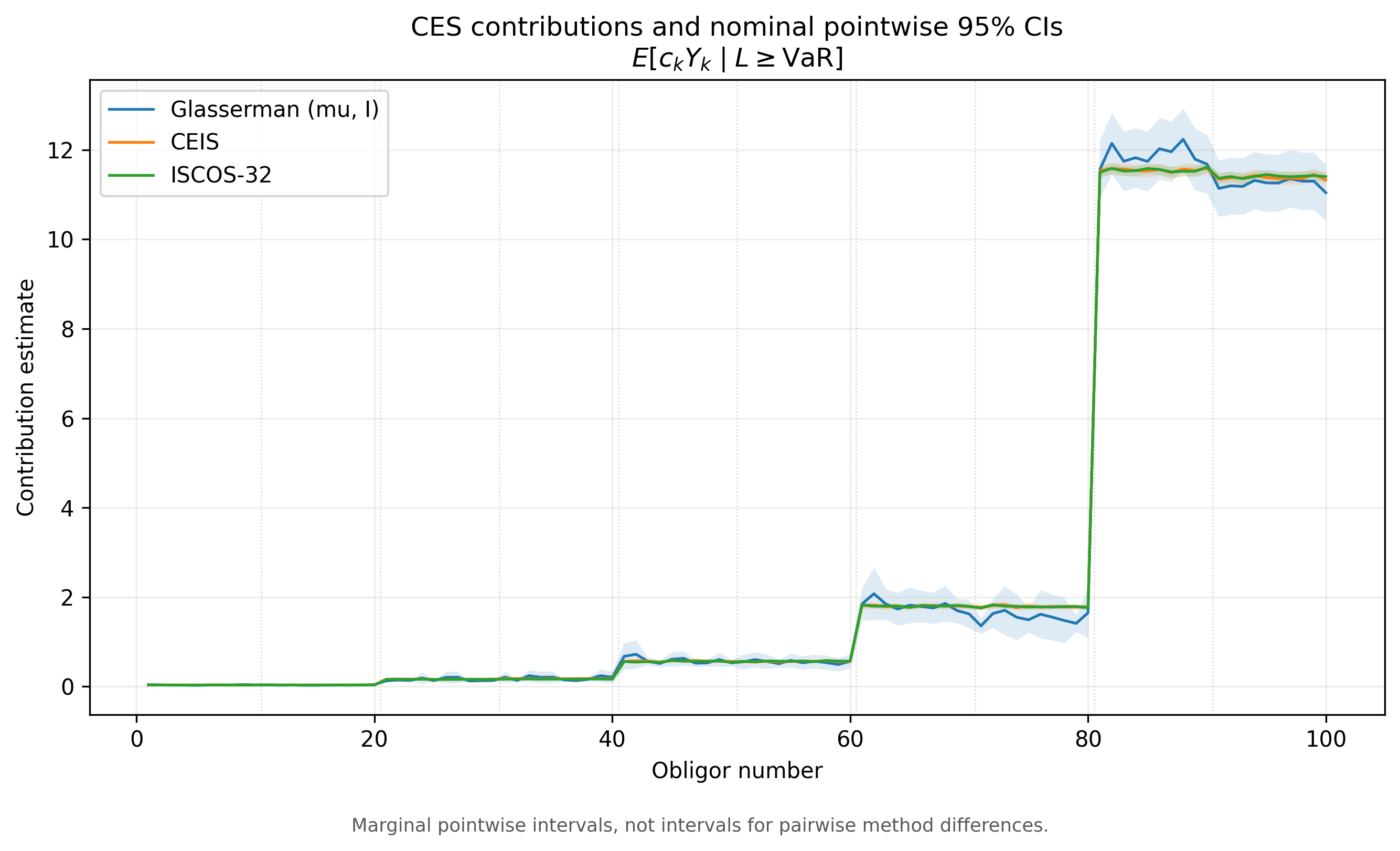}
\caption{CES contributions conditional on \(L\geq250\).}
\end{subfigure}
\caption{Obligor-level contribution estimates with nominal marginal
pointwise 95\% intervals. The methods share the threshold, budgets, and base
random-number streams within each event. The bands are not intervals for
pairwise method differences.}
\label{fig:num_contribution_estimates}
\end{figure}

\begin{figure}[p]
\centering
\begin{subfigure}[t]{0.49\textwidth}
\centering
\includegraphics[width=\textwidth]
{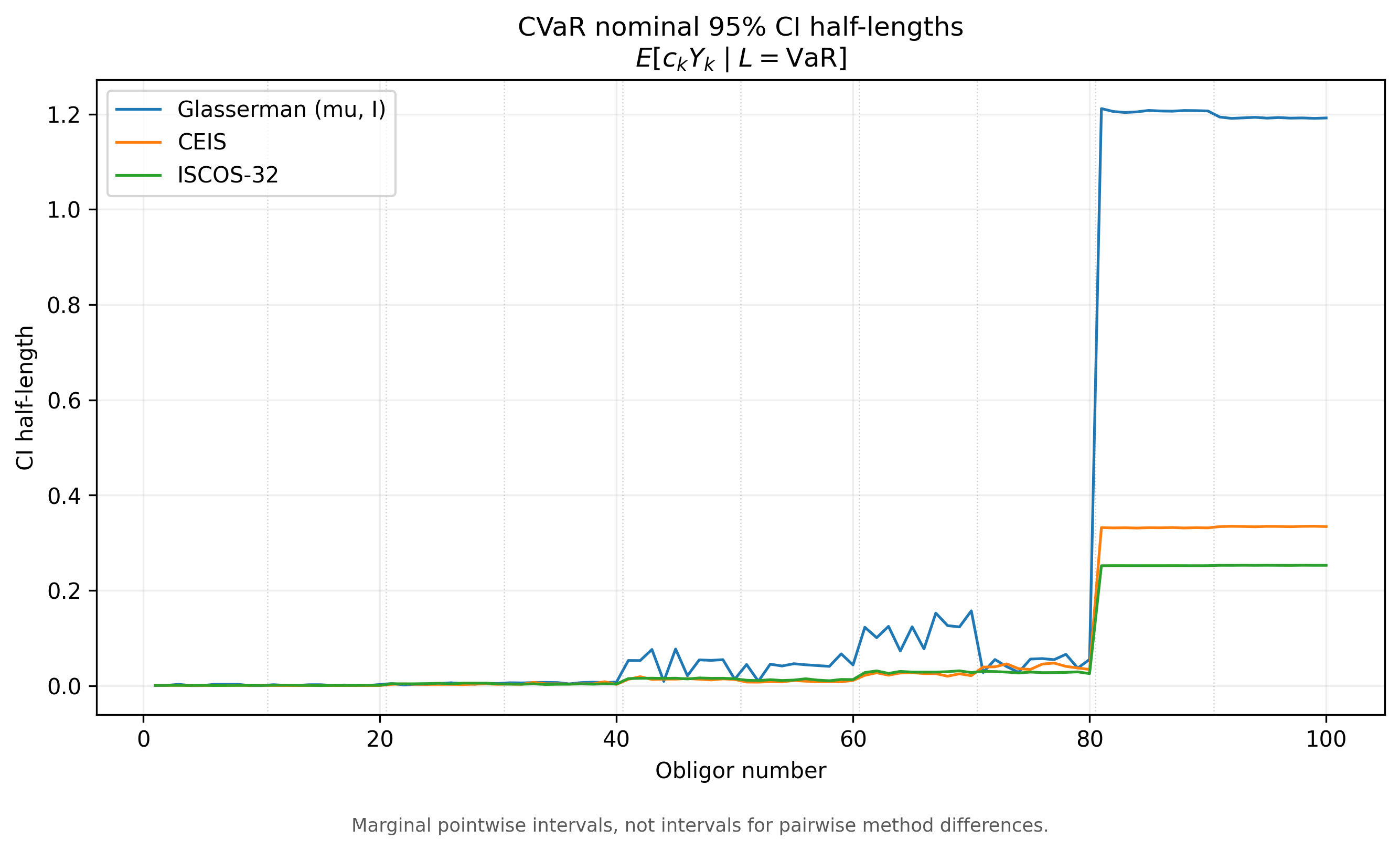}
\caption{Nominal CVaR interval half-lengths.}
\end{subfigure}
\hfill
\begin{subfigure}[t]{0.49\textwidth}
\centering
\includegraphics[width=\textwidth]
{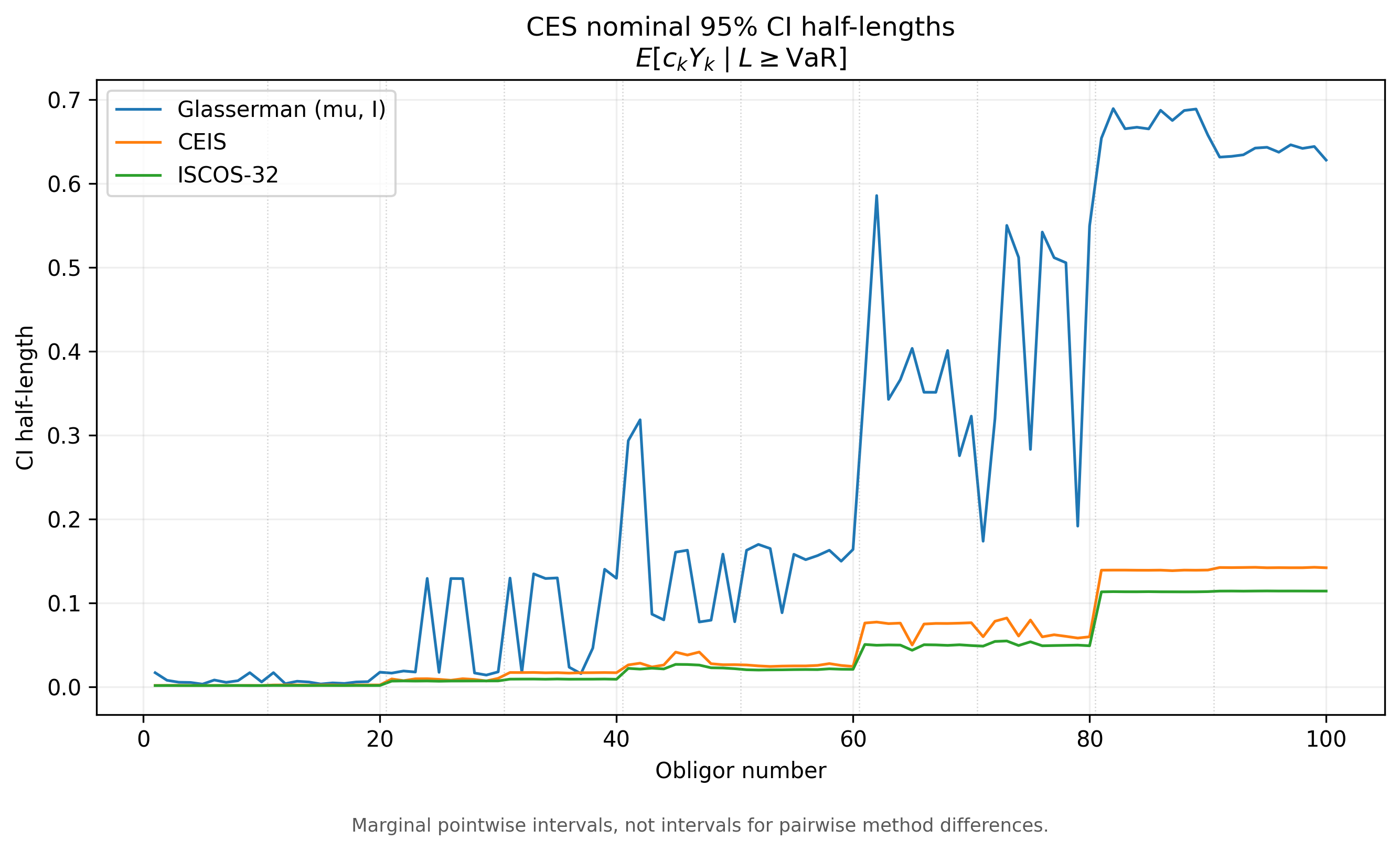}
\caption{Nominal CES interval half-lengths.}
\end{subfigure}
\caption{Pointwise interval half-lengths by obligor. The sharp increase for
obligors 81--100 reflects the two blocks with loss-at-default 25.}
\label{fig:num_ci_half_lengths}
\end{figure}

The intervals require a separate integrability qualification. The Gaussian
factor-ratio check in
\eqref{eq:num_gaussian_lr_second_moment_condition} passes for Glasserman and
ISCOS but fails for CEIS. More strongly, a direction-wise calculation finds a
CEIS covariance eigenvalue below \(1/2\) whose signed eigenvector has strictly
positive loadings for every obligor. This proves that the complete CEIS CES
importance weight has an infinite second moment in the reported configuration.
The CEIS CES bands are therefore finite-run dispersion summaries rather than
CLT-valid confidence intervals. The same sufficient divergence diagnostic
does not produce a witness for Glasserman or ISCOS, but failure to find a
witness is not a proof that every complete event-weight second moment is
finite. For the exact-level CVaR estimator, the event \(\{L=250\}\) may supply
additional decay, and a complete integrability analysis is not undertaken.
For these reasons all bands are described as nominal.

\subsection{Computational cost and limitations}
\label{subsec:num_computational_cost}

The homogeneous block structure reduces the cost of ISCOS calibration. Its
conditional characteristic function is
\begin{equation}
\varphi_{L\mid z}(\omega)
=
\prod_{g=1}^{10}
\left[1-p_g(z)+p_g(z)e^{\mathrm{i}\omega c_g}\right]^{10},
\label{eq:num_grouped_cf}
\end{equation}
so the implementation processes ten distinct Bernoulli groups rather than 100
individual obligors. The shared preliminary Monte Carlo stage takes
approximately 1.45 seconds. In the saved fixed-order timing run, the equivalent
stand-alone pipelines for all three methods are roughly 96--97 seconds; the
production CVaR and CES simulations account for almost all of each bar. The
additional filtered-COS calibration therefore has little effect on total cost
in this example.

\begin{figure}[htbp]
\centering
\includegraphics[width=0.88\textwidth]
{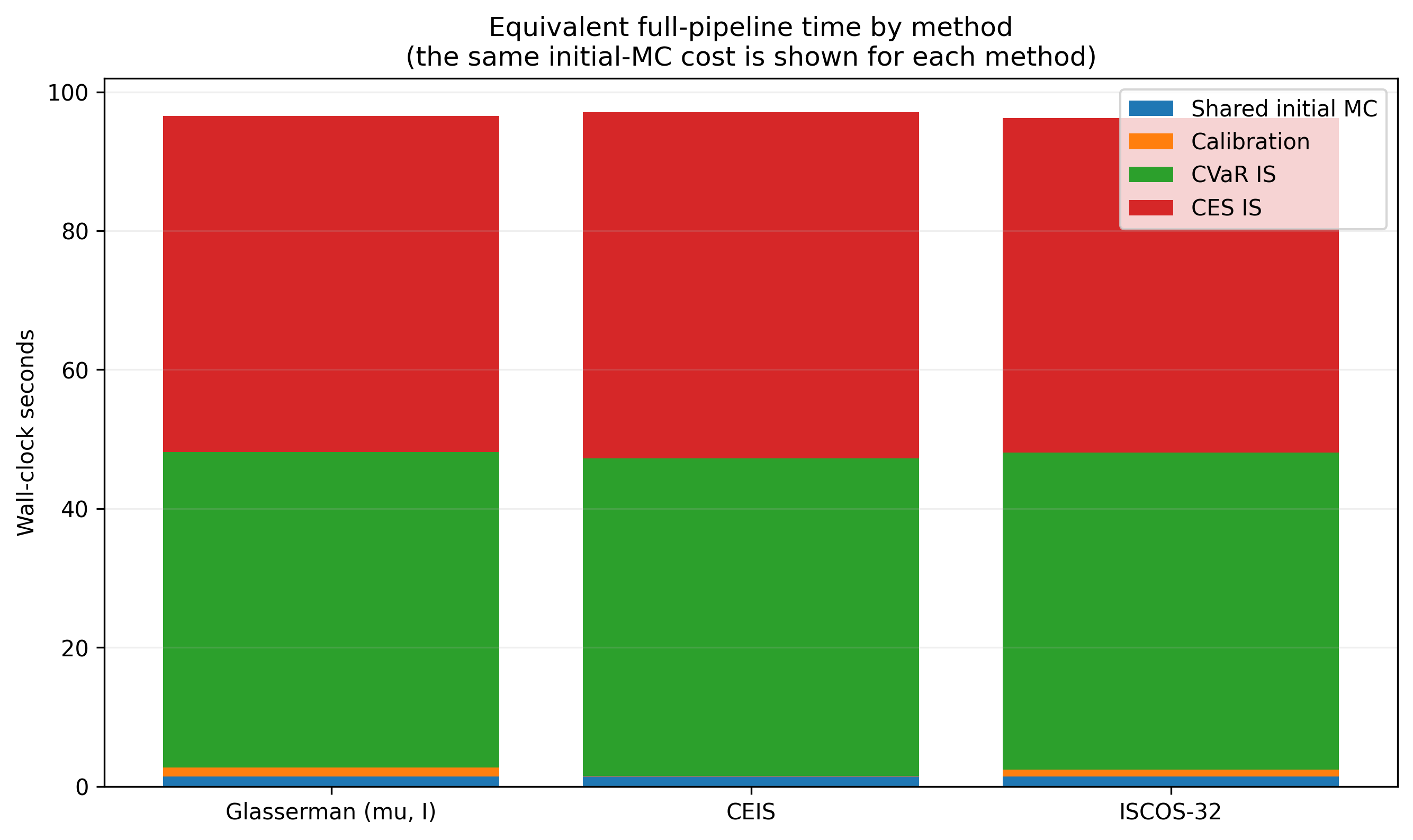}
\caption{Equivalent full-pipeline wall-clock time by method. The same shared
preliminary Monte Carlo cost is included in every bar so that each bar can be
read as a stand-alone pipeline. Plotting and output writing are excluded.}
\label{fig:num_method_timings}
\end{figure}

The timing comparison is descriptive. It consists of one sequential run in a
fixed method order, without warm-up repetitions or randomisation of that order,
and the recorded metadata do not identify the processor, memory, Python
version, or numerical-library versions. The absolute times are therefore
within-run measurements, not portable performance benchmarks.

The Gaussian experiments support three conclusions. First, the exact
conditional-convolution study confirms that the filtered COS proposal moments
approach their exact-weight counterparts as the resolution increases. Second,
under the matched pilot, filtered ISCOS-32 uses substantially more calibration
information than CEIS, is slightly closer to the exact finite-pilot CE moments,
and remains on the finite-Gaussian-ratio side of the second-moment boundary.
Third, its production weights are more balanced in this run, yielding higher
event-weight ESS and shorter average nominal intervals, with the clearer
componentwise improvement occurring for CES.

These conclusions remain limited to one pilot and one production run at a
fixed threshold. The experiment does not estimate repeated-run bias, variance,
coverage, or mean squared error; the displayed intervals exclude pilot and
calibration uncertainty; common random numbers do not convert marginal bands
into paired-difference intervals; and closeness to the CE moment target does
not imply variance optimality for every downstream ratio estimator. The
second-moment warning for CEIS CES is especially important. The results thus
support smooth conditional-probability weighting as a proposal-calibration
strategy in this benchmark, not a universal downstream variance ordering.

\section{Numerical Experiments: Student \texorpdfstring{\(t\)}{t}-Copula}
\label{sec:numerical_student_t}

This section examines the Student~\(t\)-copula implementation of the
Gaussian--inverse-Gamma product proposal derived in
Section~\ref{sec:ce_factor_space}. The experiment is deliberately matched to
the preceding Gaussian study: it uses the same eleven-factor portfolio, the
same pilot and production sample sizes, and the same definitions of the
exact-level VaR contribution and the tail conditional expected-shortfall
contribution. The comparison is between CEIS-t and ISCOS-t. A direct
Student~\(t\) analogue of the Glasserman mean-shift benchmark is not included
in the supplied implementation.

The purpose of the experiment is therefore narrower than that of the Gaussian
section. We ask whether replacing binary pilot indicators by COS conditional
tail probabilities produces a more stable product-proposal fit and more
balanced production likelihood weights when the common state contains both a
Gaussian factor vector and a heavy-tailed scale variable. The experiment is a
single matched-budget comparison rather than a COS-resolution or repeated-run
convergence study.

Throughout this section, ``CVaR contribution'' retains the convention
\[
\mathbb{E}[l_kY_k\mid L=\operatorname{VaR}_{\alpha}],
\]
whereas CES denotes
\[
\mathbb{E}[l_kY_k\mid L\geq\operatorname{VaR}_{\alpha}].
\]
All confidence intervals reported below are conditional, nominal, marginal
pointwise intervals from the production ratio estimators. They do not include
uncertainty from the pilot VaR, proposal calibration, or the finite COS
expansion, and they are not intervals for a CEIS-t minus ISCOS-t difference.

\subsection{Portfolio, Product Proposal, and Reproducible Setup}
\label{subsec:t_copula_setup}

We reuse the 100-obligor, eleven-factor portfolio described in the Gaussian
experiment. The first factor is market-wide and the remaining ten factors are
industry-specific. Each obligor has unconditional default probability
\(0.01\), and the block losses-at-default are
\[
(c_1,\ldots,c_{10})=(1,1,4,4,9,9,16,16,25,25),
\]
with ten obligors in every block. Hence \(0\leq L\leq1100\) and the loss lies
on the integer lattice.

For the Student~\(t\)-copula, the common state is
\[
U=(Z,W),
\qquad
Z\sim\mathcal{N}_{11}(0,I_{11}),
\qquad
W=\frac{\nu}{V},
\quad V\sim\chi^2_{\nu},
\]
with \(Z\) and \(W\) independent under the original measure. We set
\(\nu=4\), so that
\[
W\sim\operatorname{InvGamma}(2,2)
\]
under the shape--scale parameterisation used in this paper. Conditional on
\((Z,W)=(z,w)\), the number of defaults in block \(g\) is
\[
D_g\mid Z=z,W=w
\sim
\operatorname{Binomial}\bigl(10,p_g(z,w)\bigr),
\qquad g=1,\ldots,10,
\]
independently across blocks, where
\begin{equation}
 p_g(z,w)
 =
 \Phi\left(
 \frac{
 T_4^{-1}(0.01)/\sqrt{w}-0.3z_1-0.8z_{g+1}
 }{\sqrt{0.27}}
 \right).
\label{eq:t_num_conditional_pd}
\end{equation}
This is the lower-tail sign convention used in the theoretical sections: a
more negative factor realisation increases conditional default probabilities.
Consequently, adverse proposal mean shifts appear with negative signs in the
reported parameter estimates.

Both methods fit the independent product family
\begin{equation}
 g(z,w)
 =
 \phi_{11}(z;\mu,\Sigma)
 f_{\mathrm{IG}}(w;a_W,b_W),
\label{eq:t_num_product_proposal}
\end{equation}
where the Gaussian component matches weighted first and second moments of
\(Z\), while the inverse-Gamma component matches the weighted sufficient
statistics \(\log W\) and \(W^{-1}\). CEIS-t uses the calibration weight
\(\mathbf{1}_{\{L\geq x\}}\). ISCOS-t replaces it by the clipped EXP-4 COS
approximation of
\(\mathbb{P}(L\geq x\mid Z,W)\), using 64 modes including the zero mode.

Table~\ref{tab:t_num_setup} records the common experimental design. The two
methods share the same pilot sample and empirical VaR. In the production
stage, the implementation also shares proposal-independent base random
numbers: standard normals for the Gaussian component, uniforms transformed by
the inverse-Gamma inverse CDF, and uniforms for the conditional Bernoulli
defaults. This common-random-number construction makes the method comparison
more directly aligned, but it does not turn the marginal intervals into
confidence intervals for method differences.

\begin{table}[t]
\centering
\caption{Parameters used in the Student~\(t\)-copula comparison. The mode
count includes the zero mode.}
\label{tab:t_num_setup}
\small
\begin{tabular}{@{}ll@{}}
\toprule
Quantity & Value \\
\midrule
Degrees of freedom
& \(\nu=4\) \\
Confidence level and threshold
& \(\alpha=0.999\), \(x=\widehat{\operatorname{VaR}}_{\alpha}=504\) \\
Shared pilot sample
& \(M_0=250{,}000\), seed 42 \\
Production samples per method
& \(250{,}000\) for \(L=504\) and \(250{,}000\) for \(L\geq504\) \\
Production seeds
& exact-level: 420; tail: 421 \\
Random-number coupling
& proposal-independent inverse-CDF common random numbers \\
COS specification
& \(K=64\), EXP-4 filter \(\exp[-8(k/K)^4]\) \\
COS interval and cutoff
& \([a,b]=[-0.5,1100.5]\), \(y_x=503.5\) \\
Batch size and homogeneous groups
& 5,000; ten groups of ten obligors \\
Covariance regularisation
& ridge \(\rho=10^{-8}\) \\
Probability safeguard
& raw COS tail probabilities projected onto \([0,1]\) \\
\bottomrule
\end{tabular}
\end{table}

The production estimator changes both components of the common state and also
exponentially twists the conditional Bernoulli probabilities. Its full
likelihood ratio is therefore the product of the Gaussian factor ratio, the
inverse-Gamma scale ratio, and the conditional Bernoulli ratio. The same
production estimator and the same twisting rule are used after CEIS-t and
ISCOS-t calibration; only the fitted product proposal differs.

\subsection{Proposal Calibration}
\label{subsec:t_copula_calibration}

Table~\ref{tab:t_num_calibration} summarises the fitted proposals. The CEIS-t
pilot contains 255 realised tail hits, so its binary calibration-weight ESS is
exactly 255. The smooth ISCOS-t weights have an ESS of 353.5, an increase of
\(38.6\%\) under the same pilot budget. The two fitted Gaussian means are
similar: their Euclidean distance is \(7.66\%\) of the CEIS-t mean norm. The
relative Frobenius distance between the two covariance matrices is
\(12.61\%\). These are between-method distances, not errors, because no exact
conditional-moment benchmark is available for this Student~\(t\) run.

Both methods shift the market factor most strongly, with fitted means
\(-1.298\) for CEIS-t and \(-1.333\) for ISCOS-t. The largest industry shifts
occur in the two blocks with loss-at-default 25, where the fitted means are
approximately \((-0.757,-0.838)\) and \((-0.797,-0.841)\), respectively.
Thus, the calibrated stress direction is economically consistent with the
portfolio's lumpy exposure profile.

\begin{table}[t]
\centering
\caption{Student~\(t\)-copula proposal-calibration diagnostics. Here
\(a_W\) and \(b_W\) are the inverse-Gamma shape and scale, and
\(\kappa_2(\widehat\Sigma)\) is the spectral condition number.}
\label{tab:t_num_calibration}
\small
\begin{tabular}{@{}lrrrrrr@{}}
\toprule
Method
& \shortstack{mean\\weight}
& \shortstack{calibration\\ESS}
& \(\widehat a_W\)
& \(\widehat b_W\)
& \(\lambda_{\min}(\widehat\Sigma)\)
& \(\kappa_2(\widehat\Sigma)\) \\
\midrule
CEIS-t
& \(1.0200\times10^{-3}\)
& 255.0
& 1.9500
& 36.1980
& 0.5508
& 2.658 \\
ISCOS-t
& \(1.0689\times10^{-3}\)
& 353.5
& 1.8928
& 33.8338
& 0.5933
& 2.305 \\
\bottomrule
\end{tabular}
\end{table}

The COS diagnostics require some care. Before clipping, the ISCOS-t weights
have sample mean \(1.06625\times10^{-3}\); after clipping, the mean is
\(1.06887\times10^{-3}\), so clipping changes the mean by only
\(2.62\times10^{-6}\). Nevertheless, \(94.14\%\) of the raw values are
negative and \(94.15\%\) of the clipped values are exactly zero. This large
fraction is compatible with the fact that most original-measure pilot states
have an extremely small conditional probability of producing a loss above
504, and the raw negative values can therefore be small oscillatory
approximations around zero. It is not, however, a substitute for a pointwise
accuracy check. In contrast to the Gaussian experiment, the supplied
Student~\(t\) results contain neither an exact conditional-convolution
benchmark nor a sweep over \(K\). The 64-mode specification should therefore
be interpreted as the implemented numerical choice rather than as an
empirically validated asymptotic regime.

A second qualification concerns likelihood-ratio moments. For the product
proposal in \eqref{eq:t_num_product_proposal}, the unconditional common-state
likelihood ratio has a finite second moment only if
\begin{equation}
 \lambda_{\min}(\Sigma)>\frac12,
 \qquad
 a_W<\nu,
 \qquad
 b_W<\nu.
\label{eq:t_num_common_state_second_moment}
\end{equation}
The Gaussian condition is satisfied by both fitted covariances, and both
inverse-Gamma shapes are below \(\nu=4\). The fitted scales, however, are
\(36.198\) and \(33.834\), so the scale condition fails for both methods.
Consequently, the common-state second-moment diagnostic does not support a
central limit theorem for either set of displayed intervals. This diagnostic
is a warning rather than a complete analysis of every event-weighted
estimator: the event indicator and the conditional Bernoulli twist may provide
additional tail decay. We therefore describe all intervals below as nominal
finite-run dispersion measures.

\subsection{Risk Contributions and Nominal Pointwise Intervals}
\label{subsec:t_copula_contributions}

Table~\ref{tab:t_num_downstream} reports the main production diagnostics. The
raw proposal hit rates are nearly the same for the two methods and are
slightly lower under ISCOS-t. The event-weight ESS nevertheless increases from
257.4 to 326.7 for the exact-level event and from 22,688.8 to 43,964.5 for the
tail event. Thus, relative to CEIS-t, ISCOS-t raises the exact-event ESS by
\(26.9\%\) and the tail-event ESS by \(93.8\%\). The gain comes from more
balanced likelihood weights rather than from a higher event frequency.

\begin{table}[t]
\centering
\caption{Downstream Student~\(t\)-copula importance-sampling diagnostics.
CI half-lengths are averages over the 100 obligors and are nominal; the
second-moment qualification in the text is essential.}
\label{tab:t_num_downstream}
\small
\begin{tabular}{@{}lrrrrrr@{}}
\toprule
Method
& \shortstack{CVaR\\hit rate}
& \shortstack{CVaR\\ESS}
& \shortstack{mean CVaR\\half-length}
& \shortstack{CES\\hit rate}
& \shortstack{CES\\ESS}
& \shortstack{mean CES\\half-length} \\
\midrule
CEIS-t
& 0.8152\%
& 257.4
& 0.6404
& 77.7988\%
& 22,688.8
& 0.06891 \\
ISCOS-t
& 0.8060\%
& 326.7
& 0.5775
& 77.5692\%
& 43,964.5
& 0.05031 \\
\bottomrule
\end{tabular}
\end{table}

The two proposals give close estimates of the underlying risk quantities. The
estimated exact-level probabilities are
\(8.4364\times10^{-6}\) under CEIS-t and
\(8.5581\times10^{-6}\) under ISCOS-t. The corresponding CVaR contribution
vectors satisfy the allocation identity
\[
\sum_{k=1}^{100}\widehat{\mathrm{CVaR}}_{k,\alpha}=504
\]
up to numerical rounding. For the tail event, the estimated probabilities are
\(1.0298\times10^{-3}\) and \(1.0223\times10^{-3}\), and the threshold-based
tail means are 601.026 and 601.465. The two tail means differ by less than
\(0.1\%\).

The contribution profile follows the five exposure levels. Under CEIS-t and
ISCOS-t, the average CES contributions per obligor in the exposure groups
\(1,4,9,16,25\) are approximately
\[
(0.342,1.460,3.827,8.148,16.274)
\]
and
\[
(0.341,1.472,3.821,8.189,16.250),
\]
respectively. The CVaR contribution estimates are noisier because the
exact-level event has an event-weight ESS of only a few hundred. Even so, both
methods reproduce the same blockwise allocation pattern and place the largest
contributions in the two highest-exposure blocks.

Figure~\ref{fig:t_num_contribution_estimates} shows the point estimates and
nominal pointwise intervals. The methods are almost indistinguishable for CES,
whereas the exact-level CVaR curves display greater within-block variation.
ISCOS-t has a smaller interval half-length for 92 of the 100 CVaR components
and for all 100 CES components. Averaged over obligors, its nominal
half-length is \(9.8\%\) smaller for CVaR and \(27.0\%\) smaller for CES.

\begin{figure}[H]
\centering
\begin{subfigure}[t]{0.88\textwidth}
\centering
\includegraphics[width=\textwidth,trim=0 72bp 0 0,clip]
{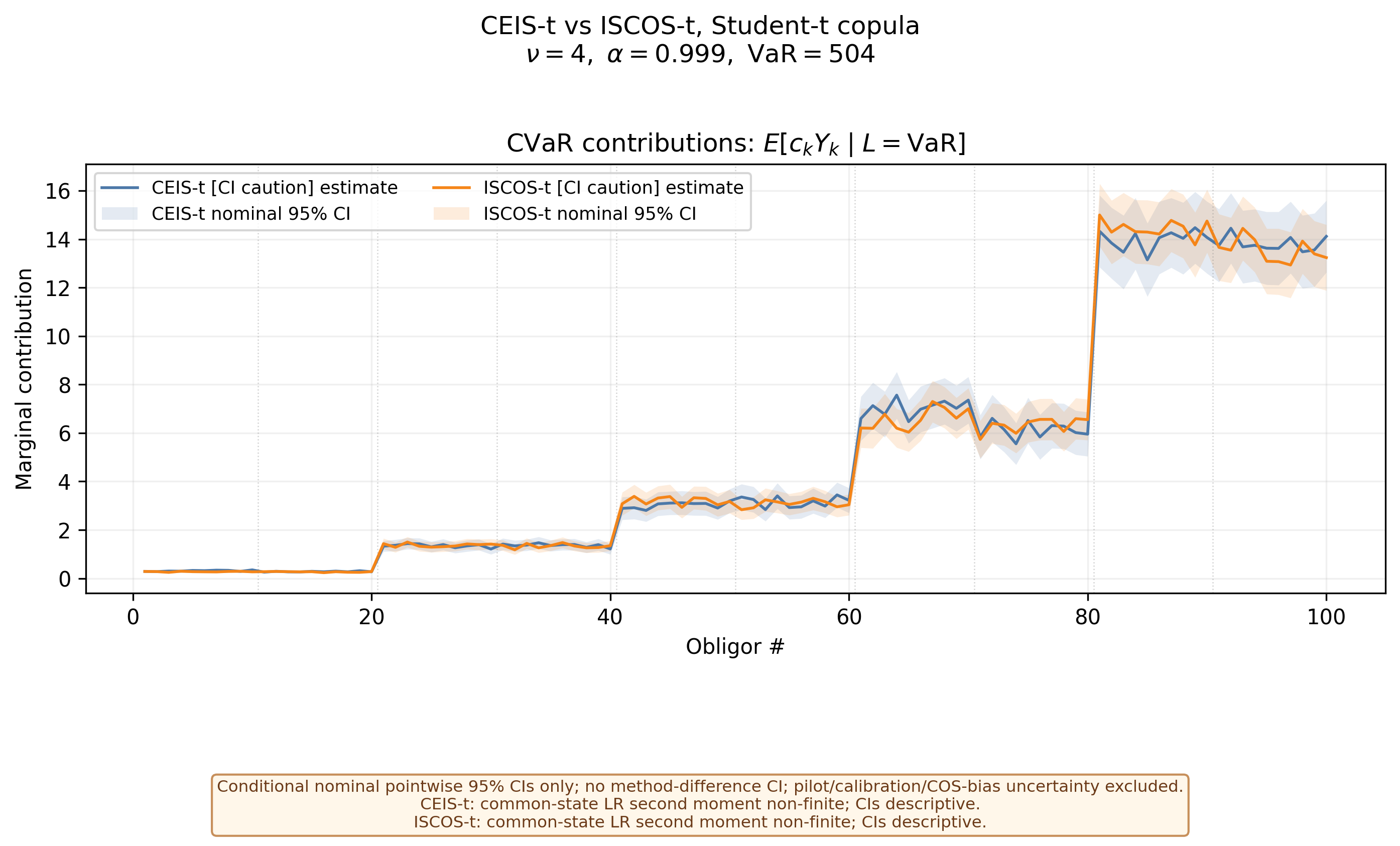}
\caption{CVaR contributions conditional on \(L=504\).}
\end{subfigure}

\vspace{0.8em}

\begin{subfigure}[t]{0.88\textwidth}
\centering
\includegraphics[width=\textwidth,trim=0 72bp 0 0,clip]
{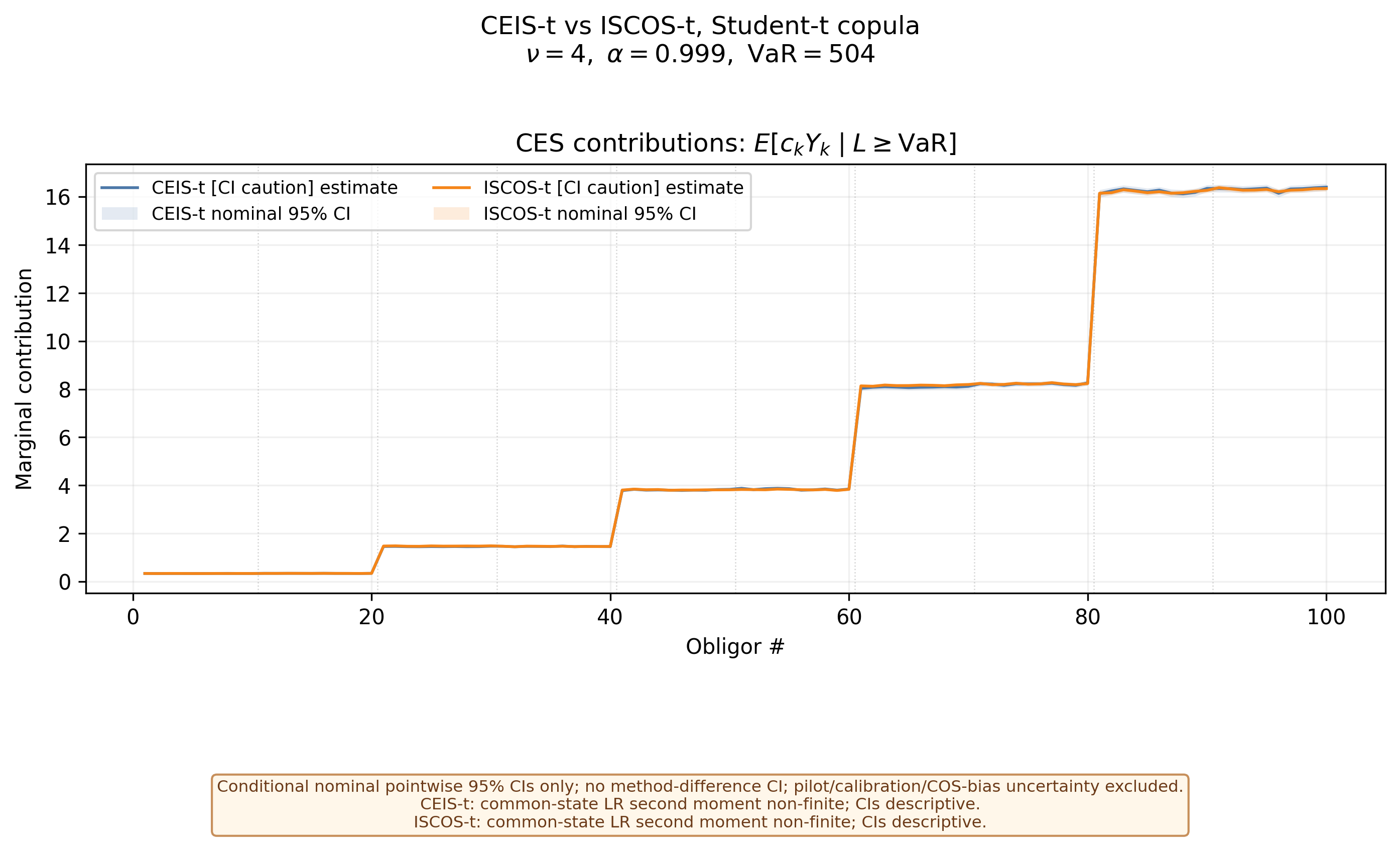}
\caption{CES contributions conditional on \(L\geq504\).}
\end{subfigure}
\caption{CEIS-t and ISCOS-t obligor-level contribution estimates with nominal
pointwise 95\% intervals. The methods share the pilot sample, VaR threshold,
production budgets, and proposal-independent base random numbers. The
intervals are descriptive because the common-state likelihood-ratio
second-moment condition fails for both proposals.}
\label{fig:t_num_contribution_estimates}
\end{figure}

Figure~\ref{fig:t_num_half_length_ratio} plots the componentwise ratio
\[
\frac{\text{CEIS-t CI half-length}}
     {\text{ISCOS-t CI half-length}}.
\]
Values above one favour ISCOS-t. The median ratios are 1.114 for CVaR and
1.396 for CES. The systematic separation of the CES curve from one is
consistent with the much larger tail-event ESS under ISCOS-t. These ratios are
comparisons of marginal nominal intervals; they are not test statistics for a
method difference.

\begin{figure}[H]
\centering
\includegraphics[width=0.92\textwidth,trim=0 70bp 0 0,clip]
{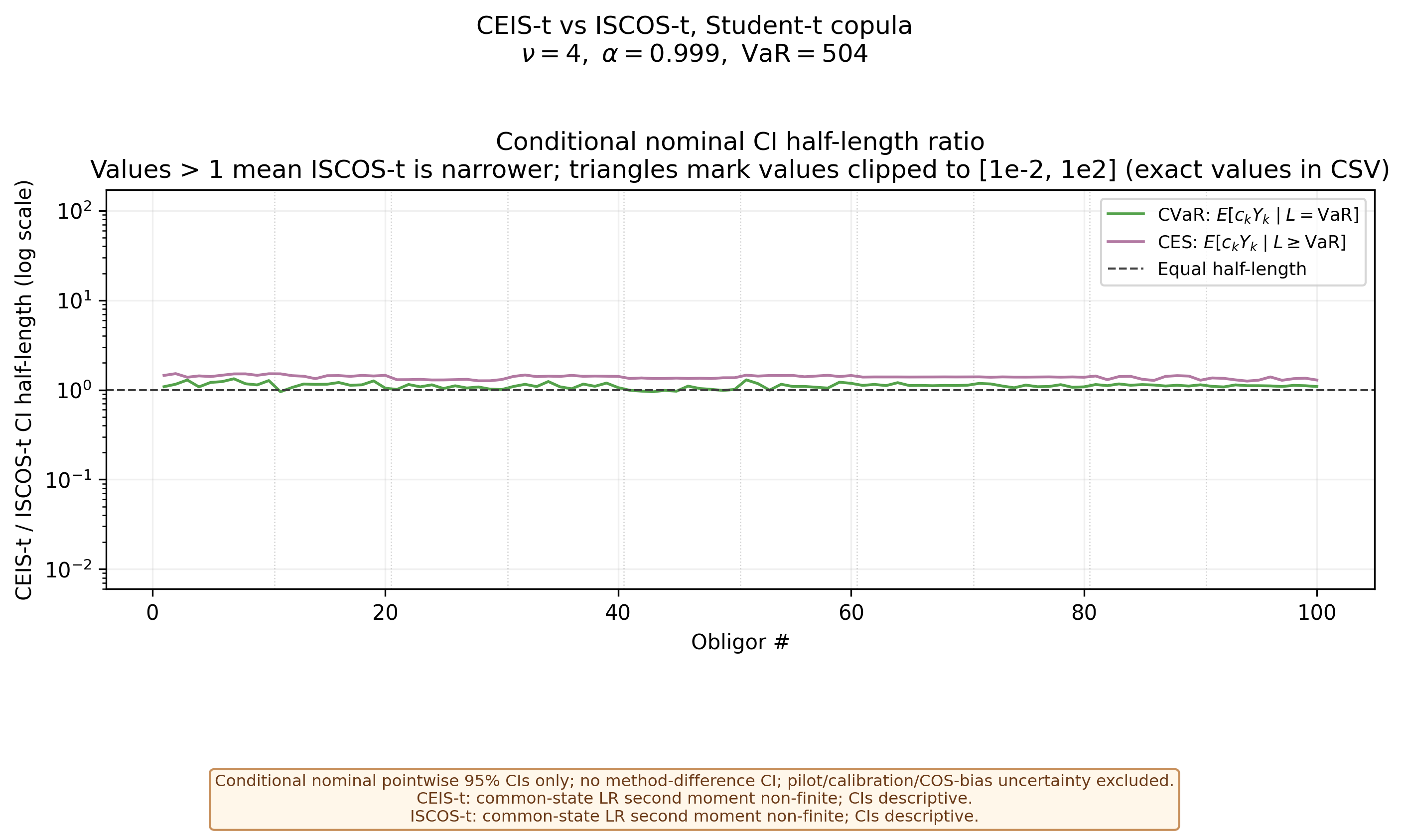}
\caption{Componentwise nominal CI half-length ratio, CEIS-t divided by
ISCOS-t. Values above one indicate a shorter interval under ISCOS-t.}
\label{fig:t_num_half_length_ratio}
\end{figure}

\subsection{Computational Cost}
\label{subsec:t_copula_cost}

The shared initial simulation takes 2.015 seconds. CEIS-t calibration is
almost costless relative to production, taking 0.101 seconds, whereas ISCOS-t
calibration takes 2.656 seconds because it evaluates the grouped conditional
characteristic function at 64 frequencies. The grouping reduces the
conditional characteristic-function product from 100 obligors to ten distinct
blocks.

The production stages dominate the total cost. The exact-level and tail runs
take 87.469 and 96.652 seconds after CEIS-t calibration, compared with 90.527
and 99.029 seconds after ISCOS-t calibration. Counting the shared initial
simulation once for each stand-alone pipeline gives equivalent total times of
186.237 seconds for CEIS-t and 194.226 seconds for ISCOS-t. ISCOS-t therefore
increases total runtime by 7.99 seconds, or \(4.3\%\), while approximately
doubling the tail-event ESS and reducing the average nominal CES interval
half-length by \(27.0\%\).

\begin{figure}[H]
\centering
\includegraphics[width=0.88\textwidth,trim=0 48bp 0 0,clip]
{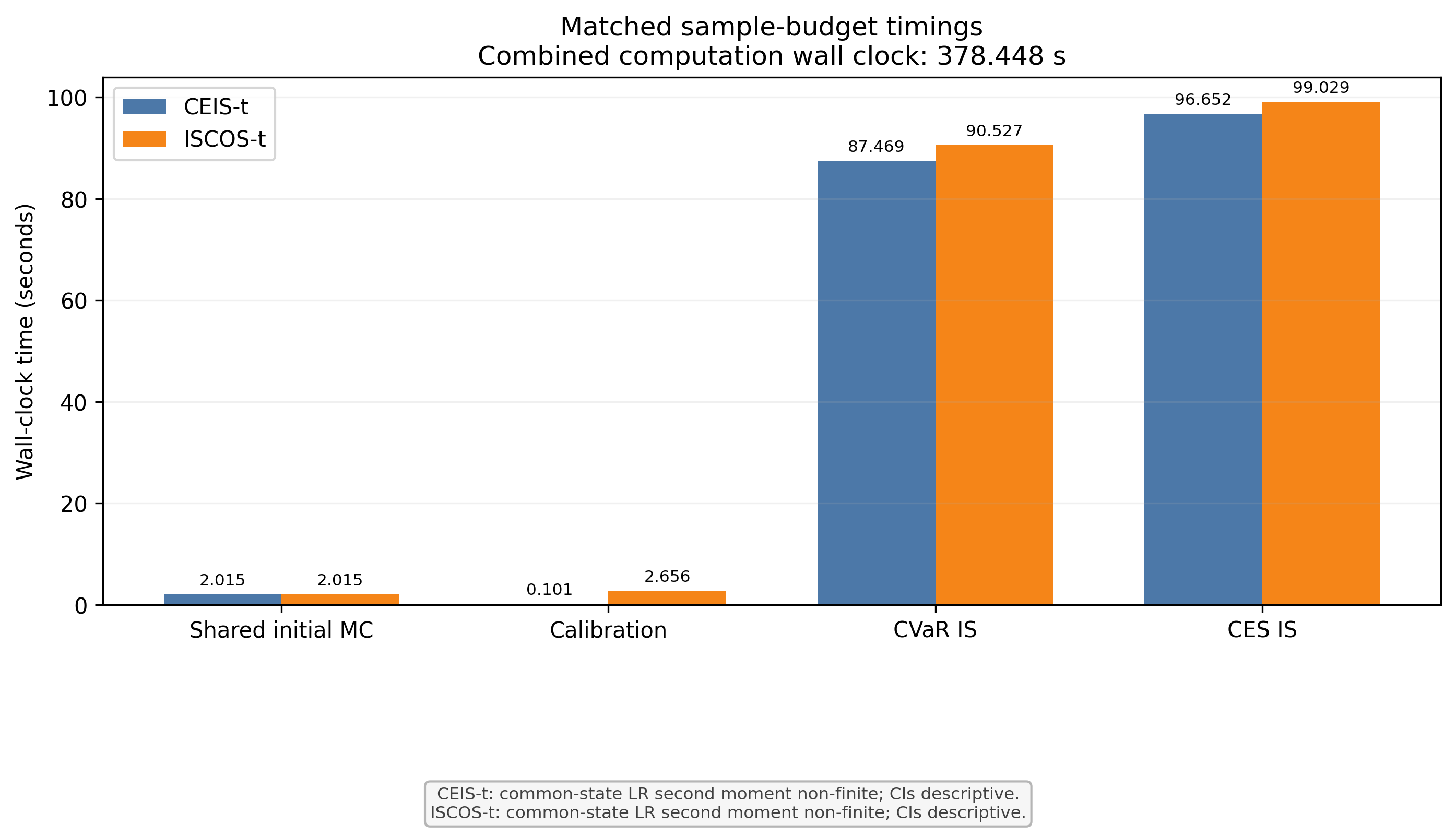}
\caption{Matched-budget wall-clock times for the Student~\(t\)-copula
comparison. The same initial Monte Carlo cost is displayed for both methods so
that each set of bars can be read as a stand-alone pipeline.}
\label{fig:t_num_timings}
\end{figure}

As in the Gaussian experiment, the recorded metadata do not identify the
processor, memory, Python version, or numerical-library versions. The absolute
times should therefore be read as within-run comparisons rather than as
portable performance benchmarks.

\subsection{Summary and Limitations of the Student \texorpdfstring{\(t\)}{t}
Experiment}
\label{subsec:t_copula_summary}

The matched Student~\(t\)-copula run gives the same broad empirical message as
the Gaussian experiment. Smooth conditional-probability weights increase the
amount of information used in proposal calibration and lead to more balanced
production likelihood weights. Under the supplied configuration, ISCOS-t
raises calibration ESS by \(38.6\%\), exact-event ESS by \(26.9\%\), and
tail-event ESS by \(93.8\%\). Its average nominal pointwise intervals are
shorter for both risk contributions, with the larger reduction occurring for
CES, while total runtime remains within \(4.3\%\) of CEIS-t.

These findings remain descriptive for four reasons. First, the experiment
contains one pilot and one production run, so it does not estimate repeated-run
bias, variance, coverage, or mean squared error. Second, there is no exact
conditional benchmark or mode sweep for the Student~\(t\) COS weights. Third,
the independent product proposal cannot represent the dependence between
\(Z\) and \(W\) induced by conditioning on a large loss. Fourth, the fitted
inverse-Gamma scales violate the common-state likelihood-ratio second-moment
condition in \eqref{eq:t_num_common_state_second_moment}. The numerical results
therefore support the practical promise of ISCOS-t, especially for tail
contributions, but they do not establish a universal variance ordering or
CLT-valid confidence intervals.

\section{Conclusion}
\label{sec:conclusion}

This paper has developed ISCOS as a cross-entropy proposal-calibration method
for rare-event simulation in multi-factor credit portfolios.  The method starts
from the common-state distribution conditional on a large portfolio loss and
projects it onto a tractable family.  This gives a Gaussian proposal for the
systematic factors and, for the Student \(t\)-copula, a
Gaussian--inverse-Gamma product proposal for the factor vector and common
scale.  The construction remains in the low-dimensional common-state space and
can be combined with conditional default twisting in the production stage.

ISCOS differs from CEIS only in the information used to fit that proposal.
CEIS uses the realised indicator \(\mathbf 1_{\{L\geq x\}}\); ISCOS uses an
approximation of
\[
q_x(U)=\mathbb P(L\geq x\mid U).
\]
Conditional independence makes the loss characteristic function available in
closed form, so a filtered COS expansion can compute this smooth weight without
an additional conditional default draw.  The half-step convention for lattice
losses evaluates the CDF away from its jump and therefore preserves the
non-strict event \(\{L\geq x\}\).

The analysis separates the benefit of conditional weighting from the error of
the finite expansion.  With exact \(q_x\), ISCOS is a Rao--Blackwellisation of
the CEIS raw sufficient-statistic estimator and has an explicitly smaller
covariance matrix in the positive-semidefinite order.  The corresponding
ordering for the fitted proposal parameters is first-order asymptotic.  If the
conditional COS approximation is uniformly \(O(K^{-p})\), the population
proposal parameters converge at the same order, while a pilot sample of size
\(M_0\) contributes the usual \(O_{\mathbb P}(M_0^{-1/2})\) term.

The numerical experiments support this calibration principle while also
showing its limits.  In the Gaussian model, exact conditional convolution
confirms that the COS-calibrated moments approach their exact-weight
counterparts as the resolution increases.  In the matched run, filtered
ISCOS-32 more than doubles the calibration ESS relative to CEIS, is slightly
closer to the finite-pilot exact-weight moments, and remains on the
finite-Gaussian-ratio side of the reported second-moment boundary.  It also
produces higher event-weight ESS and shorter average nominal intervals,
particularly for CES, at essentially the same total cost.  The componentwise
CVaR interval ordering remains mixed, and the CEIS CES second-moment warning
prevents interpreting these finite-run bands as a universal variance result.

The Student \(t\)-copula experiment reaches the same broad conclusion for a
common state containing both Gaussian factors and a heavy-tailed scale.  Under
the reported configuration, ISCOS-t increases calibration ESS by \(38.6\%\)
and tail-event ESS by \(93.8\%\), while reducing the average nominal CES
half-length by \(27.0\%\); the equivalent stand-alone runtime is \(4.3\%\)
higher than for CEIS-t.  The gain comes from more balanced event weights rather
than a higher proposal hit rate.

These results do not establish that ISCOS minimises the variance of every
final VaR or tail-contribution estimator.  The theoretical covariance order is
a calibration result, the experiments are single matched runs rather than
repeated-run studies, and the displayed intervals are nominal.  In the Student
\(t\) experiment, there is no exact conditional benchmark or COS mode sweep,
the product proposal cannot represent the dependence between \(Z\) and \(W\)
induced by the tail event, and the fitted inverse-Gamma scales fail the stated
common-state likelihood-ratio second-moment condition.  These qualifications
are essential when interpreting the interval comparisons.

Further work should therefore focus on repeated, work-normalised comparisons;
a Student \(t\) resolution study and conditional benchmark; and richer
proposals that allow the Gaussian factors to depend on the common scale.
Adaptive choices of the pilot size, COS resolution, and covariance
regularisation are also natural.  The error decomposition suggests increasing
\(K\) only until the finite-COS error is comparable with pilot sampling and
regularisation error.  More generally, ISCOS shows how conditional tail
information can stabilise proposal calibration without changing the final
importance-sampling estimator.

\bibliographystyle{plain}
\bibliography{biblio}
\end{document}